\documentclass[superscriptaddress, reprint, amsmath, amssymb, aps, prx, floatfix]{revtex4-2}

\usepackage{bm}

\usepackage{mathtools}
\usepackage{amsfonts}
\usepackage{amssymb}
\usepackage{dsfont}
\usepackage{amsthm}
\usepackage[english]{babel}
\usepackage{graphicx}
\usepackage{amsmath}
\usepackage{lipsum}
\usepackage{wrapfig}
\usepackage{float}
\usepackage{enumitem} 
\usepackage{color} 
\usepackage{hyperref}
\usepackage[normalem]{ulem}
\usepackage{url}

\useunder{\uline}{\ul}{}

\hypersetup{colorlinks = true}

\def\bra#1{\ensuremath{\mathinner{\langle{#1}|}}}
\def\ket#1{\ensuremath{\mathinner{|{#1}\rangle}}}
\DeclarePairedDelimiter\abs{\lvert}{\rvert}%
\newcommand{\norm}[1]{\left\lVert #1 \right\rVert}

\newcommand{\needcite}[1]{\textcolor{red}{[Ref needed]}}
\usepackage{qcircuit}
\newtheorem{theorem}{Theorem}
\newtheorem{definition}{Definition}
\newtheorem{lemma}{Lemma}

\newtheorem{corollary}{Corollary}[theorem]
\newtheorem{proposition}{Proposition}

\begin{document}

\title{Progress toward favorable landscapes in quantum combinatorial optimization}

\author{Juneseo Lee}
\affiliation{Department of Mathematics, Princeton University, Princeton, New Jersey 08544, USA}
\affiliation{Department of Chemistry, Princeton University, Princeton, New Jersey 08544, USA}

\author{Alicia B. Magann}
\affiliation{Department of Chemical \& Biological Engineering, Princeton University, Princeton, New Jersey 08544, USA}

\author{Herschel A. Rabitz}
\affiliation{Department of Chemistry, Princeton University, Princeton, New Jersey 08544, USA}

\author{Christian Arenz}
\affiliation{Department of Chemistry, Princeton University, Princeton, New Jersey 08544, USA}

\affiliation{School of Electrical, Computer and Energy Engineering, Arizona State University, Tempe, Arizona 85287, USA}

\date{\today}

\begin{abstract}

The performance of variational quantum algorithms relies on the success of using quantum and classical computing resources in tandem. Here, we study how these quantum and classical components interrelate. In particular, we focus on algorithms for solving the combinatorial optimization problem MaxCut, and study how the structure of the classical optimization landscape relates to the quantum circuit used to evaluate the MaxCut objective function. In order to analytically characterize the impact of quantum features on the critical points of the landscape, we consider a family of quantum circuit ans\"atze composed of mutually commuting elements. We identify multiqubit operations as a key resource, and show that overparameterization allows for obtaining favorable landscapes. Namely, we prove that an ansatz from this family containing exponentially many variational parameters yields a landscape free of local optima for generic graphs. However, we further prove that these ans\"atze do not offer superpolynomial advantages over purely classical MaxCut algorithms. We then present a series of numerical experiments illustrating that non-commutativity and entanglement are important features for improving algorithm performance.

\end{abstract}

\maketitle

\section{Introduction}
Quantum computers promise to offer computational advantages over classical computers for certain high-value tasks \cite{10.1109/SFCS.1994.365700,grover1996fast,lloyd1996universal}. However, the availability of fault-tolerant quantum computers that can achieve these speedups at meaningful scales is likely years away. In the meantime, the advent of noisy, intermediate-scale quantum (NISQ) \cite{preskill_quantum_2018} devices has inspired tremendous interest in variational quantum algorithms (VQAs), which aim to leverage the computing power of NISQ devices to solve a broad range of scientific problems, with applications spanning quantum chemistry \cite{peruzzo_variational_2013}, combinatorial optimization \cite{2014arXiv1411.4028F}, machine learning \cite{Dunjko2020nonreviewofquantum}, and linear systems \cite{2019arXiv190905820B}. VQAs function by using NISQ hardware in tandem with a classical processor \cite{McClean_2016}. At the outset, an objective function $J$ is defined that encodes the solution to a problem of interest. Then, a classical computer is used to iteratively optimize $J$. The optimization is conventionally performed over a set of parameters associated with a quantum circuit, or \emph{ansatz}, which is used to evaluate $J$ on the NISQ device. In order to achieve strong performance, the classical optimization procedure and the quantum objective function evaluation must function successfully together. As such, an understanding of the associated costs and challenges of the quantum and classical aspects of VQAs, and how they interrelate, is highly desirable. 

\begin{figure}[h!]
	\includegraphics[width=0.90\linewidth]{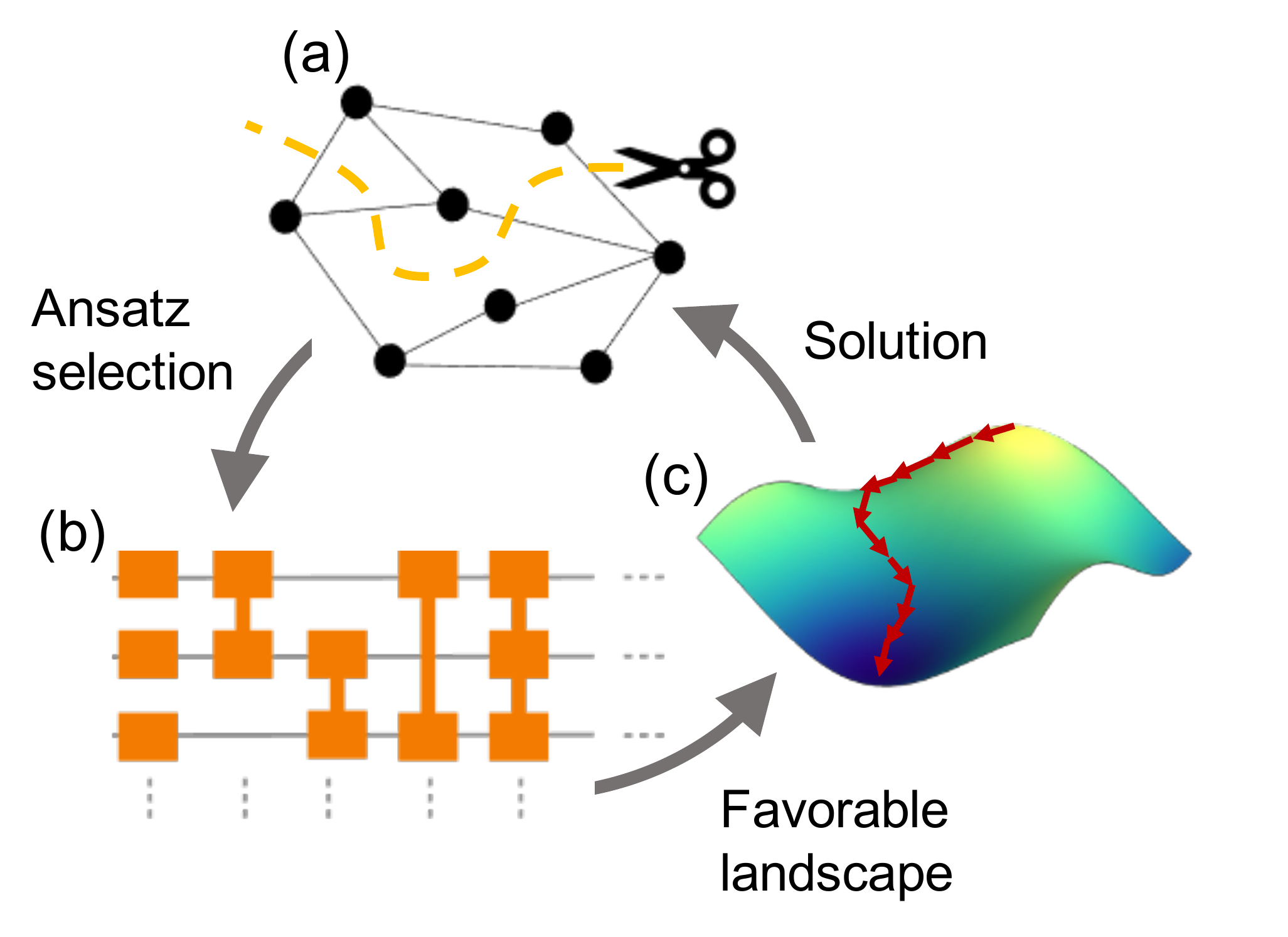}
	\caption{\label{Fig:Intro} Pictorial representation of how ansatz selection can be informed by the interplay between the problem instance and the underlying optimization landscape structure for solving the MaxCut problem. In particular, the graph structure (a) can be utilized to guide the development of an ansatz (b) for a VQA that yields favorable landscape properties (c), resulting in enhanced convergence. Here, we introduce a family of ans\"atze as a toy model that allows for analytically studying the associated landscape critical point structure, and prove that an ansatz from this family containing exponentially many variational parameters yields a landscape that is free of local optima. We go on to use this ansatz family as a starting point for exploring relations between inclusion of quantum features in ans\"atze, scalability, and VQA performance.}   
\end{figure}

To date, significant attention has been paid to the quantum component of VQAs, in an effort to develop effective strategies for evaluating the objective functions associated with different problems of interest, and a plethora of ans\"atze have been developed to target different applications \cite{peruzzo_variational_2013,2014arXiv1411.4028F,Dunjko2020nonreviewofquantum}. These application-oriented ans\"atze are often motivated by physical intuition. However, hardware-efficient ans\"atze have also been developed that are motivated by a convenient implementation on typical NISQ platforms \cite{peruzzo_variational_2013, kandala2017hardware}. Furthermore, a variety of error mitigation schemes have been developed in order to bolster the utility of ans\"atze in different settings \cite{PhysRevX.7.021050,PhysRevX.8.031027,PhysRevLett.119.180509,PhysRevA.98.062339}. In comparison to these efforts studying the quantum component of VQAs, fewer analyses have been done with respect to the classical component \cite{Grimsley2019_Adapt,Stokes2020,Koczor2020,PRXQuantum.1.020319}. The latter aspect of VQA performance can carry a significant computational cost. This circumstance arises because the classical optimization problem is non-convex in general \cite{jain2017non}, due to the fact that the quantum circuit parameters typically enter in a nonlinear manner into $J$. This can cause local optima to appear, which can render the classical search for a global optimum of $J$ prohibitively difficult \cite{BittelKliesch}.

Here, we seek to obtain a deeper understanding of these issues by analyzing the optimization landscapes, defined by the objective $J$ as a function of the quantum circuit parameters. The precise manner in which $J$ depends on the parameters is dictated by the interplay between the ansatz and the problem at hand, and consequently, the quantum and classical components of VQAs are intimately related. Thus far, numerical and theoretical observations have confirmed the presence of local optima in VQA landscapes for commonly employed ans\"atze \cite{BittelKliesch,Wierichs2020,Moll2018}. However, numerical observations have suggested that overparameterization of an ansatz can yield a more favorable landscape structure \cite{kiani2020learning, PRXQuantum.1.020319}.  We note that similar findings on overparameterization have appeared in the analysis of quantum control landscapes \cite{rabitz2004quantum,chakrabarti2007quantum,russell2017control,PRXQuantum.2.010101}, where in the latter setting, overparameterization takes the form of ``sufficient'' pulse-level control resources. In fact, VQAs can themselves be considered a form of quantum learning control experiment \cite{PhysRevLett.68.1500, PhysRevLett.118.150503, PhysRevA.101.032313, PhysRevA.102.062605, 7052406}, where the control is performed at the quantum circuit level, rather than at the conventional pulse level \cite{PRXQuantum.2.010101}. Similar findings on the effects of overparameterization have also appeared in the study of classical neural network landscapes \cite{pmlr-v119-shevchenko20a, allen2018learning, chen2019much}. Despite these observations, analytical analyses and rigorous results regarding the critical point structure of  VQA landscapes remain scarce, and a comprehensive understanding of VQA landscapes has not yet been attained. 

In this article, we make a step in this direction. Namely, we explore quantum-classical tradeoffs in VQAs by studying how the structure of the classical optimization landscape relates to the problem instance and to the quantum ansatz used to evaluate $J$. In particular, we investigate how ans\"atze can be designed to yield favorable landscapes that are free of local optima. To do so, we examine how quantum resources, such as entangling gates, can be harnessed to improve the optimization landscape structure to contain only global optima and saddle points, the latter of which often do not hinder local searches from finding global optima efficiently \cite{lee2017first,pmlr-v40-Ge15, levy2016power, pmlr-v70-jin17a}. However, as these aims are difficult to achieve in general, we focus here on applications of VQAs for solving the combinatorial optimization problem MaxCut, and consider ans\"atze that allow for analytically characterizing the critical point structure of the optimization landscape in this setting. Our general approach for this is depicted in Fig. \ref{Fig:Intro}. That is, in section \ref{sec:xansatz} we consider a family of ans\"atze with elements generated by mutually commuting $k$-body operators \cite{doi:10.1098/rspa.2008.0443, fujii2017commuting}, and use these ans\"atze as a toy model for our studies. The considered ans\"atze can be tailored to the structure of the graph under consideration, and allow for analyzing the impact of adding variational parameters in a systematic fashion. To this end, we first prove that for generic graphs on $n$ vertices, an ansatz from this family containing $2^{n-1}-1$ variational parameters yields a landscape free of local optima. We go on to explore the prospect of achieving favorable landscapes through ans\"atze with polynomially many parameters, and discuss how these findings relate to classical algorithmic capabilities. For instance, despite the inclusion of entangling operations, we show that the considered ansatz family does not offer a superpolynomial advantage over purely classical MaxCut schemes. We then numerically explore in section \ref{sec:numerics} how the algorithm performance is affected by incorporating non-commutativity into the ansatz, and compare our findings against the quantum approximate optimization algorithm (QAOA) \cite{2014arXiv1411.4028F}.

\section{\label{sec:ansatzBackground}Preliminaries}
\noindent Consider a quantum circuit 
\begin{align}
\label{eq:parameterizedCirc}
U(\bm{\theta})=\prod_{j=1}^{M}e^{-i\theta_{j}H_{j}},
\end{align}
parameterized by a set of $M$ (variational) parameters collected in the vector $\bm{\theta}=(\theta_{1},\cdots,\theta_{M})$, where $H_{j}$ are Hermitian operators. The parameterized circuit \eqref{eq:parameterizedCirc} defines an ansatz for optimizing an objective function $J(\bm{\theta})$. The objective function considered here takes the form 
\begin{align}
\label{eq:costfunction}
J(\bm{\theta})=\bra{\varphi(\bm{\theta})} H_{p}\ket{\varphi(\bm{\theta})},
\end{align}
where $H_p$ is the so-called problem Hamiltonian, and the state $\ket{\varphi(\bm{\theta})}=U(\bm{\theta})\ket{\phi}$ is created through the parameterized circuit starting from a fixed initial state $\ket{\phi}$. The goal of VQAs is then to solve the optimization problem 
\begin{align}
\label{eq:optimizationprob}
\min_{\bm{\theta}\in\mathbb R^{M}}J(\bm{\theta}). 
\end{align}
Solving (\ref{eq:optimizationprob}) is typically accomplished by iteratively searching for the parameters $\bm{\theta}$ that minimize $J$ in a hybrid quantum-classical fashion. In each iteration a quantum device is used to create the state $\ket{\varphi(\bm{\theta})}$, followed by expectation value measurements to infer the value of $J$. Then, a classical search routine is employed to determine how to update the values of $\bm{\theta}$ for subsequent iteration. This procedure is repeated until convergence is achieved. In order for this procedure to be scalable, the depth of the parameterized circuit, and the number of variational parameters $M$, should each scale at most polynomially in the number of qubits.

\subsection{Optimization landscapes of VQAs}
The ease of finding the parameter configuration that minimizes $J$ depends on the structure of the optimization landscape, given by $J$ as a function of $\bm{\theta}$. This landscape structure depends on how the components of $\bm{\theta}$ enter in the associated quantum circuit $U(\bm{\theta})$ in conjunction with the form of $H_{p}$.  In order to obtain a favorable landscape, one obvious choice is to consider ans\"atze \cite{PhysRevA.101.032308, zhang2021low} that allow for directly varying over the $\mathcal{O}(2^n)$ coefficients of $|\varphi(\bm{\theta})\rangle$ in the eigenbasis of $H_p$, assuming the eigenbasis is known, e.g., as for the MaxCut problem studied below. In this case the optimization problem is convex and constrained due to the normalization of $|\varphi(\bm{\theta})\rangle$. However, such ans\"atze lack the flexibility to systematically reduce the number of variational parameters, while still maintaining guarantees on the reachability of the ground state. Furthermore, there is no obvious mechanism for tailoring their structure to the problem instance at hand. 

In order to address these challenges, here we consider the more common case where the variational parameters $\bm{\theta}$ enter in a non-convex manner, and consequently, the associated optimization landscapes may contain local optima. In this setting, an assessment of the associated optimization landscape relies on an analysis of the set of critical points $\{\bm{\theta}^{*}\}$ at which the gradient $\nabla J(\bm{\theta})$ vanishes. For objective functions of the form \eqref{eq:costfunction}, the $j$-th component of the gradient takes the form
\begin{align}
\label{eq:gradientgeneral}
    \frac{\partial}{\partial \theta_{j}}J(\bm{\theta}) &= -i\bra{\varphi(\bm{\theta})}[H_{p}, W_{j}H_{j}W_{j}^{\dagger}]\ket{\varphi(\bm{\theta})},
\end{align}   
where $W_{j}=\prod_{k=j}^{M} e^{-i\theta_{k}H_{k}}$. Due to the form of the gradient \eqref{eq:gradientgeneral}, it is evident for each of the eigenstates of $H_{p}$, which are reachable through $U(\bm{\theta})$, that the corresponding parameter configurations constitute critical points. In addition, there could also be situations where parameter configurations that yield non-eigenstates constitute critical points. 

In order to characterize the type of critical point (e.g. saddle point, minimum, maximum), the Hessian matrix, denoted by $\nabla^{2}J(\bm{\theta})$, can be used. Using the short-hand notation $A_{j}\equiv W_{j}H_{j}W_{j}^{\dagger}$, the components of the Hessian are given by 
\begin{align}\label{eq:hess}
   \frac{\partial^{2}}{\partial \theta_{j}\partial \theta_{k}}J(\bm{\theta}) &= -\bra{\varphi(\bm{\theta})}[[H_{p},A_{j}],A_{k}]\ket{\varphi(\bm{\theta})} \nonumber \\
    &~~~~ -i\bra{\varphi(\bm{\theta})}\frac{\partial}{\partial \theta_{j}}A_{k}\ket{\varphi(\bm{\theta})}.
\end{align}
To distinguish a saddle point from a local optimum, we define a local optimum as follows:

\begin{definition}\label{def:trap}
A local optimum is a critical point that does not correspond to a global optimum or a saddle, but at which the Hessian is positive or negative semidefinite. 
\end{definition} 

As recent studies have suggested that saddle points often do not hinder gradient-based algorithms from efficiently finding global optima \cite{lee2017first,pmlr-v40-Ge15, levy2016power, pmlr-v70-jin17a}, here we focus on the question of whether it is possible to remove local optima by appropriately choosing the ansatz and the initial state $\ket{\phi}$. In particular, we explore which parameterized gates allow for such favorable landscape properties. These considerations emphasize the interplay between quantum resources and classical capabilities. Addressing them in the most general form is challenging, as it would require the ability to analytically track the dependence of $H_{p}$ and $U(\bm{\theta})$ on $\nabla J$ and $\nabla^{2}J$. As a consequence, here we focus on a particular Hamiltonian $H_{p}$ that has high practical relevance. In particular, we restrict ourselves to Ising Hamiltonians whose ground states encode solutions to the graph-partitioning problem MaxCut.

\subsection{Formulations of the MaxCut problem}
Consider a weighted, undirected graph $G=(V,E)$, where $V$ denotes the set of $n$ vertices, $E$ denotes the set of edges, and $w_{a,b}\geq 0$, with $(a,b)\in E$, denote the corresponding non-negative edge weights. The MaxCut problem then corresponds to partitioning $V$ into two subsets such that the sum of the weights belonging to the edges connecting the two subsets is maximized. More formally, for some subset of vertices $S\subset V$, whose complement is denoted by $S^{c}$, we define the cut set $\mathsf{Cut}(S)$ with respect to the partition $\{S,S^{c}\}$ by $\mathsf{Cut}(S)=\{(a,b)\in E,\, |\, a\in S, b\in S^{c}\}$. If we denote by $
 \mathsf{CutVal}(S) = \sum_{(a,b)\in \mathsf{Cut}(S)} w_{a,b}$  
the corresponding cut value, the MaxCut problem for $G$ reduces to solving
\begin{align}
\label{eq:maxcut}
    \mathsf{MaxCut}(G) = \max_{S\subset V} \mathsf{CutVal}(S).
\end{align}
We note that the MaxCut problem is equivalent to the binary quadratic program  
\begin{align}
\label{eq:binaryopt}
 &\text{minimize}~\sum_{(a,b)\in E}w_{a,b}(x_{a}x_{b}-1)/2, \nonumber \\
 &\text{subject~to}~x_{a}\in\{\pm 1\}~\text{for~every}~a\in V,
 \end{align}
 which is known to be both \textsf{NP}-hard \cite{Karp1972} and \textsf{APX}-hard \cite{Papadimitriou1991} for generic graphs $G$. As such, significant effort has been dedicated to the development of heuristics and approximation algorithms that efficiently yield high cut values. For example, the Goemans-Williamson (GW) algorithm involves the relaxation of the discrete optimization problem \eqref{eq:binaryopt} into a semidefinite program, whose initial solution is then rounded to obtain a final solution \cite{GW}.
 
It is well-known \cite{VQAReview, IsingNP} that solving the MaxCut problem is also equivalent to finding the ground state of an $n$-qubit Ising Hamiltonian 
\begin{align}
\label{eq:IsingHam}
H_{p}=\sum_{(a,b)\in E}w_{a,b}Z_{a}Z_{b},
\end{align}
where $Z_{a}=\mathds{1}\otimes\cdots \otimes Z\otimes\cdots\otimes\mathds{1}$ denotes the Pauli operator $Z=\text{diag}(1,-1)$ acting non-trivially on the $a$th qubit. In this setting, MaxCut can be formulated as the optimization problem
\begin{align}
\label{eq:maxcutquant}
\min_{\bm{\theta}\in \mathbb R^{M}}\sum_{(a,b)\in E}w_{a,b}\left(\bra{\varphi(\bm{\theta})}Z_{a}Z_{b}\ket{\varphi(\bm{\theta})}-1\right)/2.
\end{align}
If we denote the eigenstates of $H_{p}$ by $\ket{z}$ with $z\in\{0,1\}^{n}$, we see that each $\ket{z}$ corresponds to a cut of the graph as $\frac{1-\bra{z}Z_{a}Z_{b}\ket{z}}{2}\in\{0,1\}$, so it is useful to denote by $S_{z}\subset V$ the set of vertices that are assigned a ``1'' in the bitstring $z$ associated with $\ket{z}$. The equivalence $\mathsf{CutVal}(S_{z})= \mathsf{CutVal}(S_{z}^{c})$ is reflected in fact that the eigenstates of $H_{p}$ come in pairs with the same eigenvalue, or equivalently cut value, due to the $\mathbb Z_{2}$ symmetry of the Ising Hamiltonian \eqref{eq:IsingHam}. We refer to a set of $2^{n-1}-1$ vertex subsets without any of their complements as being non-symmetric. The minimization problem \eqref{eq:maxcutquant} can be considered a relaxation of the discrete MaxCut optimization problem \eqref{eq:binaryopt} to quantum states, rather than real vectors on a sphere as in the GW relaxation. It is widely hoped that varying over quantum states $\ket{\varphi(\bm{\theta})}$ will offer advantages. However, it is largely unknown which quantum features could provide such advantages.

Investigating the optimization landscape of $J(\bm{\theta})$ for the Ising Hamiltonians \eqref{eq:IsingHam} offers a potential path forward for assessing what quantum features have to offer in the context of solving the  MaxCut problem. To this end, in a recent preprint \cite{BittelKliesch} it has been shown that using an ansatz of the form
\begin{align}
\label{eq:classical}
U(\bm{\theta}) &= \prod_{j=1}^{n} e^{-i\theta_{j}X_{j}},
\end{align}
and taking the initial condition $\ket{\phi}=\ket{\bm{0}}=\ket{0}\otimes \cdots \otimes \ket{0}$ to be the highest excited state of \eqref{eq:IsingHam}, yields local optima in general, from which it is concluded that the classical optimization problem is \textsf{NP}-hard. In Eq. (\ref{eq:classical}), $X_{j}$ denotes the Pauli operator $X=\left(\begin{matrix}0&1\\1&0\end{matrix}\right)$ that acts non-trivially on the $j$th qubit. Since the eigenstates of the Ising Hamiltonian take the form $\ket{z}$ with $z\in\{0,1\}^{n}$, the classical ansatz \eqref{eq:classical} can be interpreted as continuously flipping qubits. We note that even though $\exp(-i\theta_{j}X_{j})$ allows for coherent superpositions of qubit states of the form $\cos(\theta_{j})\ket{0}-i\sin(\theta_{j})\ket{1}$, as $H_{p}$ is diagonal in the computational basis, such coherent superpositions do not give any advantage over convex combinations. That is, in both cases the objective function takes the form, 
\begin{align}
\label{eq:classicalost}
J(\bm{\theta})=\sum_{(a,b)\in E} w_{a,b}\cdot \cos(2\theta_{a})\cos(2\theta_{b}), 
\end{align}   
whose optimization landscape provably contains local optima depending on the graph structure. As such, we refer to \eqref{eq:classical} as a ``classical'' ansatz. Furthermore, given that the relative phase of each qubit state does not change $J$, even an ansatz consisting of generic local $\text{SU}(2)$ operations applied to each qubit will produce an objective function of the form \eqref{eq:classicalost}. Consequently, starting from $\ket{\bm{0}}$, the use of generic local operations alone does not yield favorable optimization landscapes. With this in mind, below we explore the effect of incorporating entangling gates.

\section{\label{sec:xansatz} A family of ans\"atze for removing local optima}
Looking beyond the classical ansatz \eqref{eq:classicalost}, we proceed by including ansatz elements that are created by $k$-body operators $\prod_{i\in S}X_{i}$ acting non-trivially on a subset $S\subset V$ of $|S|=k$ qubits. This leads to a family of ans\"atze, which we refer to as $\mathbb X$-ans\"atze.

\begin{definition}\label{simpleDef} An $\mathbb X$-ansatz is of the form 
\begin{align}
\label{eq:XAnsatz}
U(\bm{\theta})=\prod_{j=1}^{M}e^{-i\theta_{j}H_{j}},~~~~H_{j}=\prod_{i\in S_{j}}X_{i},
\end{align}
where $\mathcal A=\{S_{j}\}$ with $S_{j}\subset V$ being a collection of vertex subsets that the ansatz elements non-trivially act on. 
\end{definition}
We remark that due to commutativity of its elements the set $\mathcal A$ uniquely defines the ansatz \eqref{eq:XAnsatz}, and that \eqref{eq:classical} is contained in the family of $\mathbb{X}$-ans\"atze through $\mathcal A=\Big\{\{1\},\{2\},\cdots,\{n\}\Big\}$. 
As the application of each $H_{j}$ in \eqref{eq:XAnsatz} on $\ket{\bm{0}}$ has the effect of creating a cut, varying over $\bm{\theta}$ in a given $\mathbb X$-ansatz can be interpreted as varying continuously over cuts.  Furthermore, these $\mathbb{X}$-ans\"atze have strong ties to instantaneous quantum polynomial-time (IQP) circuits \cite{doi:10.1098/rspa.2008.0443, fujii2017commuting}. In this regard, it is interesting to note that the ability to calculate $J$ for a given $\mathbb{X}$-ansatz efficiently on a classical computer is related to the ability to determine whether barren plateaus are present \cite{mcclean2018barren, PhysRevA.102.042207, JColes}. For further details, we refer to Appendix \ref{App:Barren}.

Since the ansatz elements in \eqref{eq:XAnsatz} mutually commute, the components of the gradient \eqref{eq:gradientgeneral} take a particularly simple form 
\begin{align}
    \frac{\partial}{\partial \theta_{j}}J(\bm{\theta}) &= -i\bra{\varphi(\bm{\theta})}[H_{p}, H_{j}]\ket{\varphi(\bm{\theta})},
\end{align}
while the elements of the Hessian \eqref{eq:hess} are given by 
\begin{align}
   \frac{\partial^{2}}{\partial \theta_{j}\partial \theta_{k}}J(\bm{\theta}) =
    -\bra{\varphi(\bm{\theta})}[[H_{p},H_{j}],H_{k}]\ket{\varphi(\bm{\theta})} \nonumber ,
\end{align}
which enables the optimization landscape to be analyzed analytically. 
We proceed by fixing $\ket{\phi}=\ket{\bm{0}}$. Writing out the objective function explicitly and utilizing techniques from \cite{Bolis1980} for degenerate critical points, at which the Hessian is not invertible \cite{jain2017non}, allows for establishing the following lemma, whose proof is given in Appendix \ref{App:Lemma1Proof}.

\begin{lemma}\label{theoremNonEigen}
Given an $\mathbb X$-ansatz, any critical point of $J(\bm{\theta})$ not corresponding to an eigenstate of $H_{p}$ is a saddle.
\end{lemma}
Given the goal of understanding the optimization landscape critical point structure, and importantly, the presence of local optima in this landscape, the importance of this lemma is that it allows us to focus completely on critical points with parameter configurations $\bm{\theta}_{E}^{*}$ that correspond to eigenstates of $H_{p}$, as all other critical points are saddle points. Since $\bra{z}[[H_{p},H_{j}],H_{k}]\ket{z}=0$ for all $j\neq k$  and all $z\in\{0,1\}^{n}$ we immediately have that at these parameter configurations $\bm{\theta}_{E}^{*}$, the Hessian is diagonal, with elements given by  
\begin{align}
\label{eq:HessianDiag}
\frac{\partial^{2}}{\partial \theta_{j}^{2}}J(\bm{\theta})|_{\bm{\theta}=\bm{\theta}_{E}^{*}}=-2(J(\bm{\theta}_{E}^{*})-\bra{\varphi(\bm{\theta}_{E}^{*})}H_{j}H_{p}H_{j}\ket{\varphi(\bm{\theta}_{E}^{*})}).
\end{align}
Instead of considering $J$ as a function of $\bm{\theta}_{E}^{*}$ we can also think of $J$ as dependent on the set $S_{z}$ that corresponds to the eigenstate created. In this case we write $J\{S_{z}\}$, where we denote by $S_{0}$ and $S_{g}$ the vertex sets corresponding to highest excited and ground states, respectively. The second term in \eqref{eq:HessianDiag} describes the expectation value of $H_{p}$ with respect to an eigenstate $H_{j}\ket{z}$. The corresponding vertex set can be described by the symmetric difference of the set $S_{j}$ corresponding to $H_{j}$, describing which vertices are flipped, and the set $S_{z}$, describing the assignment of ones in $\ket{z}$. More formally, if we introduce the symmetric difference of two sets $A$ and $B$ as 
\begin{align}
A\oplus B=A\cup B- A\cap B,
\end{align}
we can express the diagonal elements of the Hessian as 
\begin{align}\label{eq:HessianDiagForm}
\frac{\partial^{2}}{\partial \theta_{j}^{2}}J(\bm{\theta})|_{\bm{\theta}=\bm{\theta}_{E}^{*}}=2(J\{S_{z}\oplus S_{j}\}-J\{S_{z}\}).
\end{align} 
We observe that for $S_{z}$ to be a local minimum, the condition
\begin{align}\label{eq:condMin}
J\{S_{z}\}\leq J\{S_{z}\oplus S_{j}\},~~~\forall S_{j}\in\mathcal A, 
\end{align} 
has to be satisfied, while for $S_{z}$ to be a local maximum we need 
\begin{align}\label{eq:condMax}
J\{S_{z}\}\geq J\{S_{z}\oplus S_{j}\},~~~\forall S_{j}\in\mathcal A,  
\end{align} 
to hold. Together with Lemma 1 this allows for establishing the following theorem. 

\begin{theorem}\label{corollaryFull}
The optimization landscape associated with an $\mathbb X$-ansatz for which $\mathcal A$ contains all non-symmetric $2^{n-1}-1$ vertex subsets exhibits no local optima. 
\end{theorem}

\begin{proof} We prove Theorem 1 by contradiction. By Lemma \eqref{theoremNonEigen} we need only consider local optima corresponding to eigenstates. Thus, assume there exists a local minimum with vertex set $S_{z}$. By the assumption that all non-symmetric vertex subsets are contained in $\mathcal A$, we can pick $S_{j}=S_{z}\oplus S_{g}$. Using properties of the symmetric difference we then have 
\begin{align}
J\{S_{z}\}&\leq J\{S_{z}\oplus (S_{z}\oplus S_{g})\}\nonumber\\
&=J\{S_{g}\},
\end{align}
which contradicts that by definition $J\{S_{g}\}<J\{S_{z}\}$. Analogously, a local maximum $S_{z}$ satisfying \eqref{eq:condMax} would contradict $J\{S_{0}\}>J\{S_{z}\}$, which completes the proof.  
\end{proof}

Theorem 1 shows that an $\mathbb X$-ansatz with $M=2^{n-1}-1$ variational parameters yields an optimization landscape whose critical points consists of global optima and saddle points only. It also shows that local optima occurring in the classical ansatz \eqref{eq:classical} vanish when ansatz elements that contain $k$-body entangling operators are added.

While Theorem 1 holds for any graph, it requires exponentially many classical parameters $\theta_{j}$ and is therefore  not scalable. We now consider whether there exist graphs for which an $\mathbb X$-ansatz with polynomially many parameters can be sufficient to obtain an optimization landscape that is free from local optima. 

We begin by considering the example of an Ising chain with nearest-neighbor interactions, described by the Hamiltonian 
\begin{align}
H_{p}=\sum_{j=1}^{n-1}w_{j,j+1}Z_{j}Z_{j+1}.
\end{align}
Depending on the weights $\omega_{j,j+1}$, the classical ansatz \eqref{eq:IsingHam} yields local optima. However, an $\mathbb X$-ansatz described by a path given by $\mathcal A=\Big\{\{1\},\{1,2\},\cdots,\{1,2,\cdots,n-1\}\Big\}$ does allow for turning all local optima into saddle points. To see this, note that from \eqref{eq:HessianDiag} we have that the diagonal elements of the Hessian at the $\ket{z}$ critical points are given by
\begin{align}
\frac{\partial^{2}}{\partial \theta_{j}^{2}}J(\bm{\theta})|_{\bm{\theta}=\bm{\theta}_{E}^{*}}=-4w_{j,j+1}\bra{z}Z_{j}Z_{j+1}\ket{z}.
\end{align}
Together with Lemma 1, we can then conclude that the only critical points at which the Hessian is positive (negative) semidefinite are global minima (maxima). Consequently, for the Ising chain, an optimization landscape free from local optima can be obtained with an $\mathbb X$-ansatz consisting of $n-1$ variational parameters and with a circuit depth polynomial in $n$. To see the latter, we note that in addition to the circuit containing only linearly many ansatz elements, each ansatz element can itself be implemented efficiently with standard universal quantum gate sets \cite{tacchino2020quantum}. Furthermore, it is straightforward to generalize this conclusion to chains with periodic boundary conditions (i.e., to all connected 2-regular graphs). In this latter case, an $\mathbb X$-ansatz described by $n$ paths, each of length $n-1$ but starting at a different vertex, gives an optimization landscape exhibiting global optima and saddle points only, using $\mathcal O(n^2)$ variational parameters. 

These examples illustrate that including quantum features in the ansatz, here in the form of $k$-body entangling operators, can allow for obtaining a favorable landscape while maintaining scalability. We now consider whether scalable $\mathbb X$-ans\"atze can yield a landscape free from local optima for other classes of graphs, where MaxCut is nontrivial. Mathematically, this translates into the question of whether for a given graph $G$, an $\mathbb X$-ansatz with $|\mathcal A|=\text{poly}(n)$ exists so that conditions \eqref{eq:condMin} and \eqref{eq:condMax} can only be satisfied at the global optima.

\begin{theorem}\label{theorem:randomized}
For any graph $G$ and an $\mathbb{X}$-ansatz $\mathcal{A}$ with size $\abs{\mathcal{A}}$, there exists a purely classical algorithm that has the same solution set as the set of local optima satisfying conditions \eqref{eq:condMin} and \eqref{eq:condMax}.
\end{theorem}
\begin{proof}
Consider the purely classical algorithm for solving \textsf{MaxCut}, shown in Fig. \ref{Fig:theorem2} and outlined in the associated figure caption. The condition that an output of the classical algorithm is a cut in which the \textsf{CutVal} cannot be increased any further by a flip of a single set of vertices $S_{k}$ is precisely condition \eqref{eq:condMin}.
Changing ``increase" to ``decrease" in the algorithm immediately yields those that satisfy \eqref{eq:condMax}, although for the purposes of MaxCut this set is irrelevant.
This completes the proof.
\end{proof}
\begin{figure}[h!]
	\includegraphics[width=0.9\linewidth]{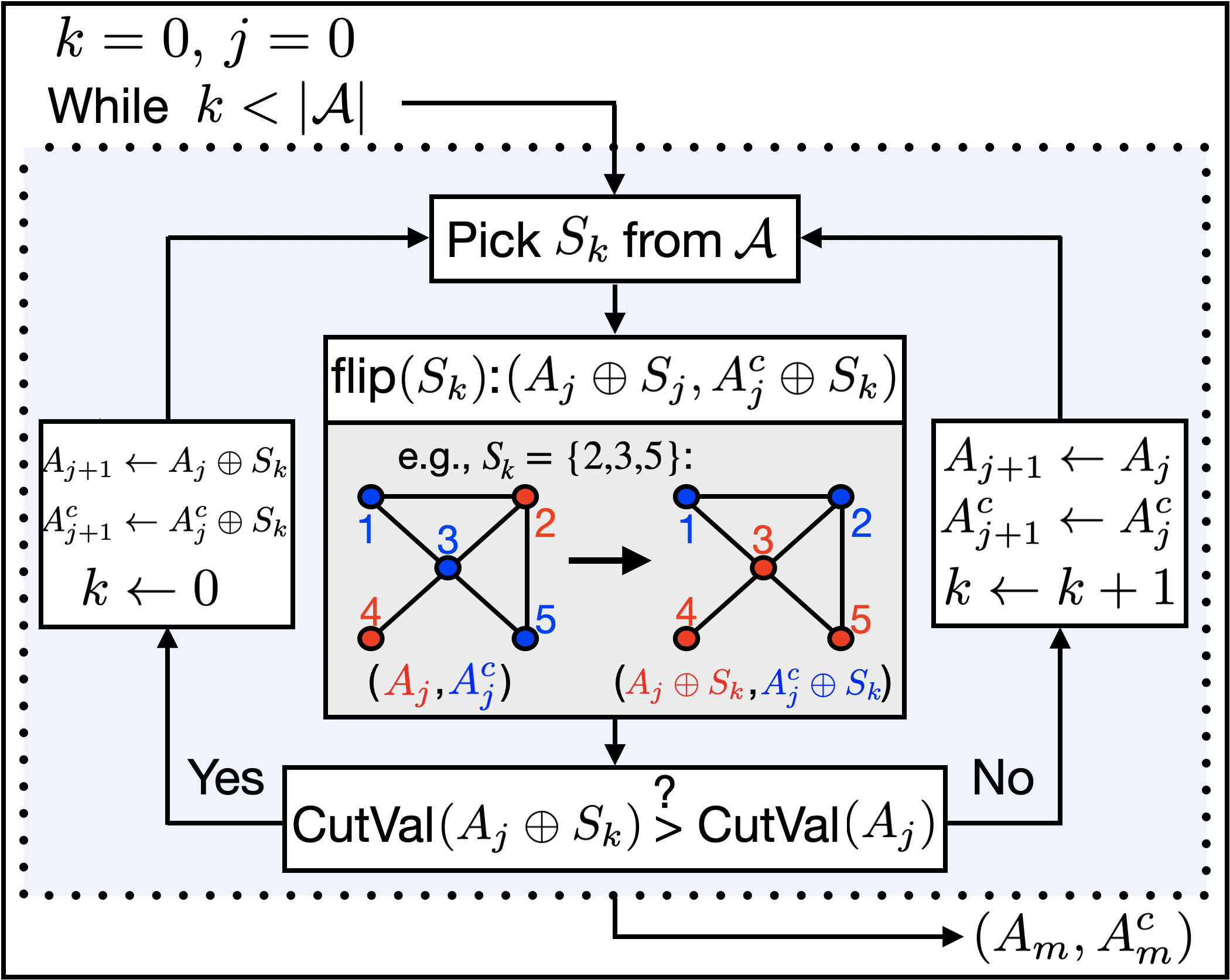}
	\caption{\label{Fig:theorem2} Purely classical algorithm for solving MaxCut. Start with a random bipartition $(A_{0},A^c_{0})$ of the vertices $V$. Then, iteratively construct new bipartitions $(A_{1},A_{1}^c), (A_{2},A_{2}^c), \ldots, (A_{m},A_{m}^c)$. At each iteration $j=1,\cdots,m$, pick an ansatz element $S_{k}\in \mathcal{A}$. Beginning with an ansatz element labeled by $k=0$, implement the $\mathsf{flip}(S_{k})$ operation by flipping the assignments of each vertex in $S_{k}$ to obtain $A_{j}\oplus S_{k}$ and $A_{j}^c\oplus S_{k}$. If $\mathsf{flip}(S_{k})$ increases the \textsf{CutVal} such that $\mathsf{CutVal}(A_{j}\oplus S_{k})>\mathsf{CutVal}(A_{j})$ (or, equivalently, decreases the objective function $J$), then increment $j$ and set $A_{j+1}=A_{j}\oplus S_{k}$ and $A_{j+1}^c=A_{j}^c\oplus S_{k}$ and return to $k=0$. Otherwise, increment $k$ and repeat the \textsf{flip} operation until $\mathsf{CutVal}(A_{j}\oplus S_{k})>\mathsf{CutVal}(A_{j})$ is satisfied. Continue this procedure until a bipartition $(A_{m},A_{m}^c)$ is reached such that $\mathsf{CutVal}(A_{m})$ cannot be increased any further by a single \textsf{flip}.}
\end{figure}

We remark that the manner in which a particular $S_{k}$ is picked at each step in the algorithm shown in Fig. \ref{Fig:theorem2} is irrelevant.
Choosing $S_{k}$ uniformly at random from $\mathcal{A}$ yields a randomized algorithm, whereas iteratively testing each $S_{k}\in \mathcal{A}$ and choosing the largest-increasing $S_{k}$ at each step yields a greedy algorithm akin to the classical 0.5-approximation scheme for MaxCut \cite{Kahruman2007}. 

This immediately yields the following corollary:
\begin{corollary}\label{cor:TrapFreeClassical}
Unless \textsf{P}$=$\textsf{NP}, there does not exist a class of $\mathbb{X}$-ans\"atze with $\abs{\mathcal{A}}=poly(n)$ that yields a landscape free from local optima for generic graphs.
\end{corollary}
\begin{proof}
If there exists a landscape free from local optima, then the only cut that satisfies either \eqref{eq:condMin} or \eqref{eq:condMax} is the one corresponding to the global minimum (namely, the maximum cut) or global maximum. 
By Theorem \eqref{theorem:randomized}, if there exists an $\mathbb{X}$-ansatz $\mathcal{A}$ with size $\abs{\mathcal{A}}=poly(n)$, there also exists a purely classical greedy algorithm that would always converge to the global minimum as well.
Notice that for an unweighted graph the maximum cut value is at most $\abs{E}$, so the purely classical algorithm converges in at most $\abs{\mathcal{A}}\cdot \abs{E}$ steps, thus solving unweighted MaxCut with polynomial cost.
Since unweighted MaxCut is \textsf{NP}-hard, and so is MaxCut for arbitrary graphs (see \cite{Kahruman2007} for the extension of greedy algorithms to weighted graphs), this shows that unless \textsf{P}$=$\textsf{NP}, there does not exist a class of $\mathbb{X}$-ans\"atze with $\abs{\mathcal{A}}=poly(n)$ that generically yields landscapes consisting of global optima and saddle points only. 
\end{proof}

This relation can be extended even further, by considering approximation schemes for MaxCut, aimed at achieving an approximation ratio $\alpha = \textsf{CutVal}(S)/\textsf{MaxCut}(G)$ for some $S\subset V$. Then, assuming the existence of an algorithm that can escape saddle points:

\begin{corollary}\label{cor:approxSaddle}
Given a fixed approximation ratio $\alpha$, and any $\mathbb{X}$-ansatz with $|\mathcal{A}|=poly(n)$, even an algorithm that can escape saddle points cannot provide a superpolynomial advantage over a purely-classical $\alpha$-approximation scheme for MaxCut. 
\end{corollary}
\begin{proof}
Analogously as to the proofs of Theorem 2 and Corollary 2.1, the key idea is that the solution set among the local optima that satisfy conditions \eqref{eq:condMin} or \eqref{eq:condMax} is indifferent to each such potential solution in a gradient algorithm, even one that can escape saddle points.
Similarly, the classical approximation scheme presented in the proof of Theorem \eqref{theorem:randomized} is also indifferent to each of the potential solutions, and for a fixed $\alpha$ it converges in at most $\frac{2\cdot \abs{\mathcal{A}}}{1-\alpha}$ steps, which is polynomial in $n$ if $\abs{\mathcal{A}}=poly(n)$.
As such, any provable approximation ratios $\alpha$ given by each of the algorithms are identical. Consequently, no superpolynomial advantage exists.
\end{proof}

This result begs the question: what is required for a quantum advantage in this setting? In the following section, we assess the role of non-commutativity through a series of numerical experiments.

\section{\label{sec:numerics}Numerical experiments}

To systematically assess how including $k$-body elements in an ansatz affects the structure of the underlying optimization landscape, we define the $k$-body depth of an $\mathbb X$-ansatz as $D=\max_{S_{j}\in \mathcal A}|S_{j}|$. 
Here, we remark that while ``width" or ``span" may be more apt descriptors of this quantity with respect to ansatz elements acting on a graph, we choose to refer to it as a ``depth" in order to serve as a proxy for the more common circuit depth complexity, as depicted in Figure \eqref{Fig:Xcirc}(a).
Given this, we first aim to explore how the structure of the optimization landscape changes when the $k$-body depth is successively increased, towards an ansatz containing all $k$-body operators, which according to Theorem \eqref{corollaryFull} yields a landscape free from local optima. We then proceed by introducing non-commutative elements in the $\mathbb X$-ansatz, and investigate whether such extensions exhibit faster convergence to better approximation ratios. Finally, we compare the $\mathbb X$-ansatz and its non-commutative variants against QAOA.

We focus our numerical analyses on complete graphs $K_{n}$ with random positive edge weights, for which MaxCut is known to be \textsf{NP}-hard \cite{Karp1972}.  
In each numerical experiment we solve \eqref{eq:maxcutquant} for $K_{n}$ with $w_{a,b}$ chosen uniformly randomly from $[0,5]$, and utilize the first-order gradient Broyden–Fletcher–Goldfarb–Shanno (BFGS) algorithm with a randomly chosen initial parameter configuration $\bm{\theta}$. Details regarding the hyperparameters used can be found in Appendix \ref{App:ClassicalNumerics}. In each run we calculate the approximation ratio $\alpha$, given as the ratio between the actual MaxCut value $\mathsf{MaxCut}(K_{n})$, obtained from exact diagonalization of $H_{p}$, and the cut value obtained from solving \eqref{eq:maxcutquant} using BFGS. The curves in the figures below show the average taken over 100 realizations and the shaded areas show the corresponding standard deviation.

\subsection{Dependence on the $k$-body depth}
We begin by investigating how the approximation ratio that is obtained changes when the $k$-body depth in the $\mathbb X$-ansatz is increased. In particular, we consider the $\mathbb X$-ansatz schematically represented in Fig. \ref{Fig:Xcirc}(a) where increasing the $k$-body depth by one is achieved by adding $n \choose k$ $k$-body operators. 
\begin{figure}[!h]
	\includegraphics[width=0.90\linewidth]{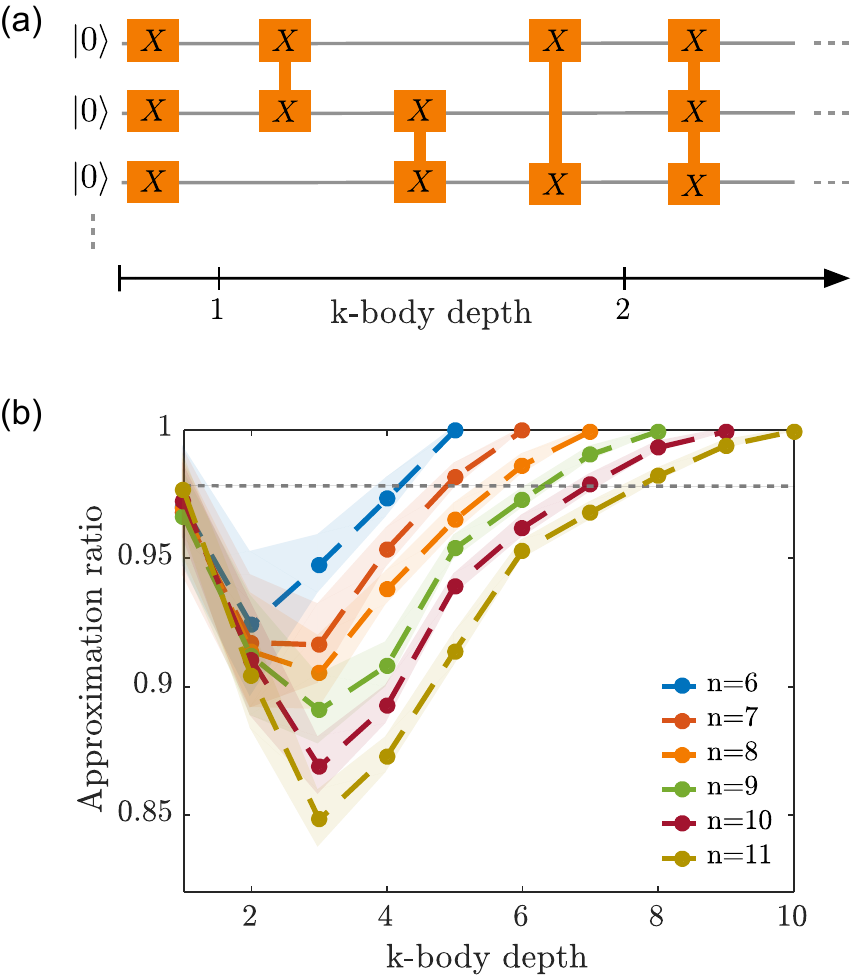}
	\caption{The performance of the $\mathbb X$-ansatz, whose circuit diagram is shown in (a), for solving MaxCut on complete graphs is shown in (b). In (a) boxes represent variational ansatz elements. Connected boxes represent entangling ansatz elements that are generated by $k$-body $X$ operators acting non-trivially on $k$ ``boxed'' qubits. In (b) the approximation ratio is shown as a function of the $k$-body depth for different problem sizes $n$. The circles correspond to the average taken over 100 randomly chosen graph instances and BFGS initial conditions. The shaded areas show the corresponding standard deviations. The grey dashed line indicates the threshold for when better approximation ratios than the classical ansatz \eqref{eq:classical} are achieved.}   
	\label{Fig:Xcirc}
\end{figure}
That is, for a fixed $k$-body depth $D$ the ansatz consists of $M=\sum_{k=1}^{D} {n\choose k}$ ansatz elements, noting that $D=1$ corresponds to the classical ansatz \eqref{eq:classical} with $M=n$ local rotations. We further note that for $D=n$ we have $M=2^{n}-1$, so that for $D=n-1$ the number of variational parameters $M=2^{n}-2$ scales exponentially in the number of qubits $n$. The approximation ratio as a function of the $k$-body depth is shown in Fig. \ref{Fig:Xcirc}(b) for different $n$. We first observe that the classical ansatz performs very well as approximation ratios $\geq 0.95$ are achieved. Numerical simulations shown in Appendix \ref{App:ClassicalNumerics} indeed confirm that solving \eqref{eq:maxcutquant} using BFGS for graphs with a large vertex degree yield approximation ratios that are slightly better than the ones obtained from the GW algorithm. However, from Fig. \ref{Fig:Xcirc}(b) we also see that adding quantumness in the form of $2$-body entangling terms does not increase the approximation ratio. Instead, performance drops until a sufficiently large $k$-body depth is reached (here $>3$). This behavior suggests that simply adding additional $2$-body terms to the classical ansatz does not automatically make the optimization landscape structure more favorable, as a drop in the approximation ratio indicates that randomly initialized first-order gradient algorithms are in this case more prone to get stuck in local optima or saddle points. 
However, increasing the $k$-body depth even further allows for obtaining better approximation ratios than the classical ansatz, which is indicated by a grey dashed line. We further observe in Fig. \ref{Fig:Xcirc}(b) that when $D=n-1$, i.e, exponentially many variational parameters are used, average approximation ratios of $\geq0.99$ with standard deviations  $\leq 10^{-6}$ are achieved. This behavior is in line with Theorem 1, as an $\mathbb X$-circuit with exponentially many parameters yields a landscape free from local optima, while the appearance of saddle points does not seem to affect the performance of BFGS. However, we remark here that according to Theorem 1, a landscape not exhibiting local optima is already obtained at a lower $k$-body depth $D > n/2$, as all non-symmetric vertex subsets are then contained in $\mathcal A$. It is interesting to note that in this case, smaller approximation ratios correspond to runs converging to degenerate saddle points. 

One way to justify this behavior disappearing when we increase the $k$-body depth from $n/2$ to $n-1$ is to consider that at depth $n-1$, each parameter is effectively included twice (for any ansatz element $S_{k}$ with parameter $\theta_{k}$, there exists its complementary element $S_{j}=S_{k}^{c}$ for some $j$, with parameter $\theta_{j}$). Since the critical point conditions are equivalent for $\theta_{k}$ and $\theta_{j}$, the probability of all pairs satisfying the conditions (and thus yielding a critical point) at depth $n-1$ is squared relative to the probability that each of the $2^{n-1}-1$ values at depth $n/2$ satisfies them.
Thus, the probability of observing this phenomenon at depth $n-1$ is significantly lower than at depth $n/2$. 
  
The results shown in Fig. \ref{Fig:Xcirc}(b) suggest that the $\mathbb X$-ansatz with sufficiently many $k$-body terms performs better than the classical ansatz \eqref{eq:classical}. However, according to Corollary \eqref{cor:approxSaddle}, there is also a purely classical strategy to achieve the same approximation ratios. A natural next question to consider is whether introducing non-commutativity in the $\mathbb X$-ansatz will improve performance.

\subsection{Assessing the role of non-commutativity}
We proceed by introducing non-commutative ansatz elements into the $\mathbb X$-ansatz. We consider ans\"atze of the form 
\begin{align}
U(\bm{\theta})=\prod_{j}e^{-i\tilde{\theta}_{j}\tilde{H}_{j}}e^{-i\theta_{j}H_{j}},
\end{align}
where $H_{j}\in \left\{\prod_{i\in S}X_{i}\,|\,S\in \mathcal A\right\}$ are the generators of the $\mathbb X$-ansatz elements determined by $\mathcal A$ and non-commutativity is introduced through the Hermitian operators $\tilde{H}_{j}$. As schematically represented in Fig. \ref{Fig:NonCommute}(a) and (b), we consider the case where the $k$-body depth is increased by including ansatz elements generated by Pauli $Z$ operators between the elements of the $\mathbb X$-ansatz in the last section, which we refer to as $\mathbb {XZ}$-ans\"atze. 
\begin{figure}[h!]
	\includegraphics[width=0.90\linewidth]{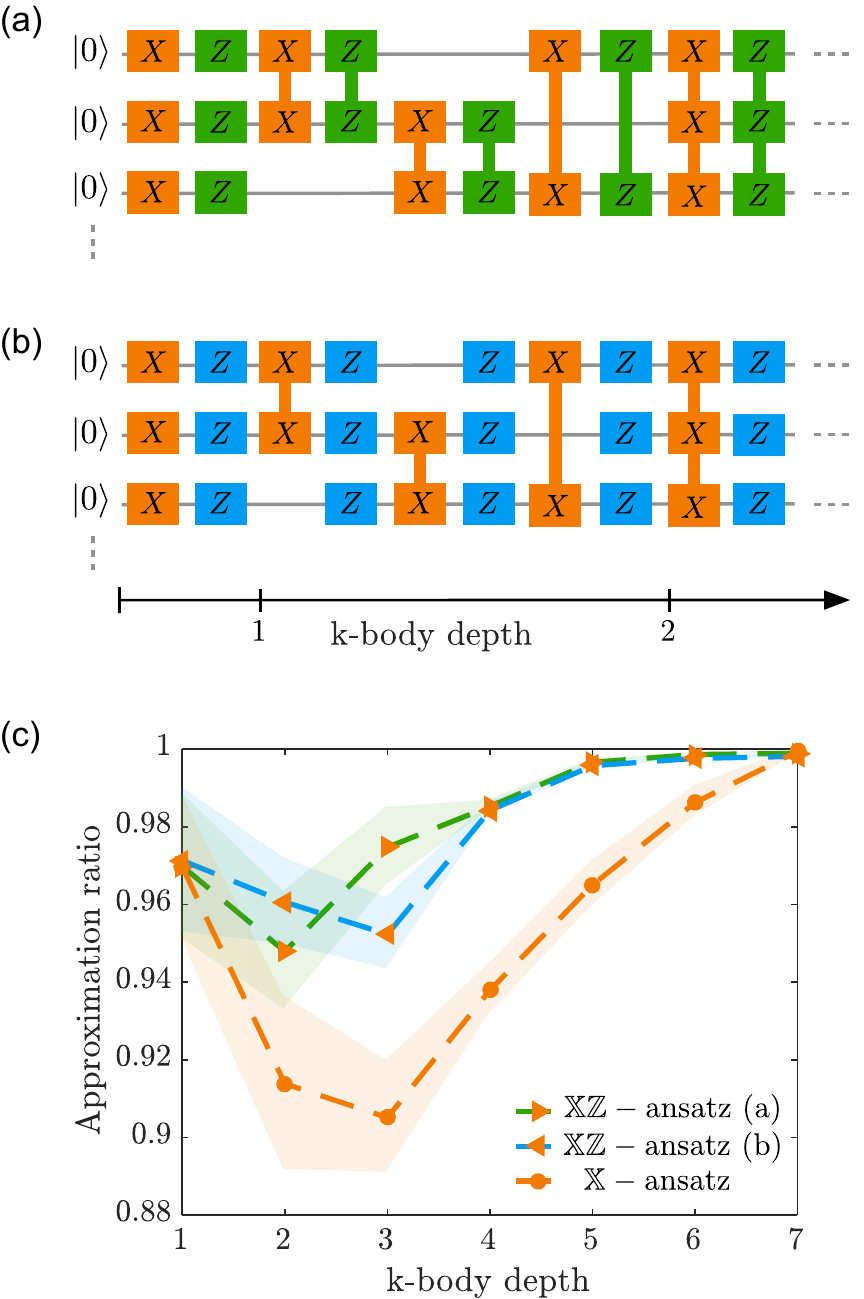}
	\caption{The role of non-commutativity is assessed by introducing variational ansatz elements generated by Pauli $Z$ operators into the $\mathbb X$-ansatz. This is achieved by alternating between ansatz elements generated by $k$-body $X$ (orange) and (a) $k$-body $Z$ (green) operators and (b) single-qubit $Z$ operators (blue). In (c), the approximation ratio is shown as a function of the $k$-body depth for $n=8$ qubits. The circles/triangles correspond to the average taken over 100 randomly chosen complete graphs and BFGS initial conditions. The shaded areas show the corresponding standard deviations.}   
	\label{Fig:NonCommute}
\end{figure}
We treat two different cases. Namely, in (a)  we consider $\tilde{H}_{k}\in \left\{\prod_{i\in S}Z_{i}\,|\,S\in \mathcal A\right\}$ while in (b) $\tilde{H}_{k}=\sum_{i=1}^{n}Z_{i}$ for all $k$. As such, now  
the number of variational parameters is increased by $2{n\choose k}$ when the $k$-body depth is increased by one. The results are shown for $n=8$ in Fig. \ref{Fig:NonCommute}(c).

We first observe that both $\mathbb {XZ}$-ans\"atze yield faster convergence than the $\mathbb X$-ansatz, which is not surprising as the number of variational parameters has doubled. 
However, it is interesting to observe that the performance between (a) and (b) differs only slightly; at a $k$-body depth of $4$, both $\mathbb {XZ}$-ans\"atze yield approximation ratios of $\approx 0.98$ while the $\mathbb X$-ansatz achieves $\approx 0.94$, which is even lower than for the classical ansatz \eqref{eq:classical}.  

Another way of introducing non-commutativity is repeating a given structure, which consists of ansatz elements that do not mutually commute. A standard example for such an ansatz is QAOA \cite{2014arXiv1411.4028F}.

\subsection{Comparison with QAOA and initial state dependence}
QAOA is a VQA designed to solve combinatorial optimization problems that was developed in 2014 \cite{2014arXiv1411.4028F} and has since inspired many works \cite{otterbach2017unsupervised, qiang2018large, willsch2020benchmarking,abrams2019implementation,bengtsson2019quantum,Pagano2020,Harrigan2021}. QAOA aims to generate approximate solutions to combinatorial optimization problems such as MaxCut through an ansatz of the form 
\begin{align}
\label{eq:QAOA}
U(\bm{\theta})=\prod_{j}e^{-i\tilde{\theta}_{j}H_{p}}e^{-i\theta_{j}H_{X}}, 
\end{align}
where $H_{X}=\sum_{j=1}^{n}X_{j}$. In QAOA the initial state is given by $\ket{\phi}=\ket{\bm{+}}$ where $\ket{\bm{+}}=\ket{+}\otimes\cdots\otimes \ket{+}$ with $\ket{+}$ being an eigenstate of $X$, so that for sufficiently many alternations between $H_{p}$ and $H_{X}$ as in \eqref{eq:QAOA} a ground state of $H_{p}$ is reachable. In contrast to the $\mathbb X$- and $\mathbb {XZ}$-ans\"atze in the last section, for which the ground state is reachable at a $k$-body depth of $D=1$, a sufficiently large circuit depth is needed in QAOA before the ground state can be created.

\begin{figure}[h!]
	\includegraphics[width=0.8\linewidth]{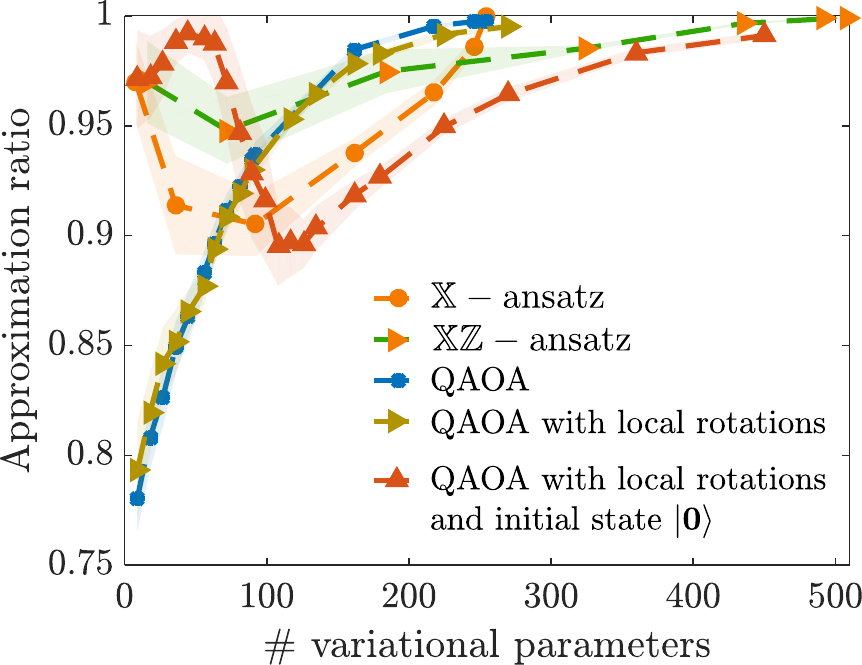}
	\caption{A comparison of different versions of QAOA with the $\mathbb X$- and $\mathbb {XZ}$-ans\"atze for $n=8$ qubits is shown. The circles/triangles correspond to the average approximation ratio taken over 100 randomly chosen complete graphs and BFGS initial conditions. The shaded areas show the corresponding standard deviations.}    
	\label{fig:QAOAcomp}
\end{figure}

In Fig. \ref{fig:QAOAcomp} we compare the approximation ratios obtained from QAOA (blue) with the $\mathbb X$- and $\mathbb {XZ}$-ansatz (orange and orange/green, respectively) represented in Fig. \ref{Fig:Xcirc}(a) and Fig. \ref{Fig:NonCommute}(a) as a function of the number of variational parameters for $n=8$ qubits. We note that since Fig. \ref{fig:QAOAcomp} compares the performance obtained using the $\mathbb{X}$- and $\mathbb{XZ}$-ans\"atze against the performance of QAOA and its variants, we plot the results as a function of the number of variational parameters, rather than the $k$-body depth, as the latter is not a relevant quantity in QAOA. However, for the former $\mathbb{X}$- and $\mathbb{XZ}$-ans\"atze, the number of variational parameters serves as a proxy for the $k$-body depth, as per Figs. \ref{Fig:Xcirc} and \ref{Fig:NonCommute}.

We first observe that for a small number of variational parameters, QAOA performs worse than the $\mathbb X$- and the $\mathbb {XZ}$-ans\"atze. This behavior can be traced back to the fact that for a small number of variational parameters (and accordingly, a shallow circuit), the ground state of $H_{p}$ may not be reachable through the ansatz \eqref{eq:QAOA}. In contrast, as the classical ansatz \eqref{eq:classical} is contained in the $\mathbb X$- and the $\mathbb {XZ}$-ans\"atze, for $M=n$ variational parameters the ground state is already reachable for a small number of classical parameters, which explains why the $\mathbb X$- and the $\mathbb {XZ}$-ans\"atze perform better in this regime than QAOA. However, we also see from Fig. \ref{fig:QAOAcomp} that with increasing numbers of variational parameters, QAOA outperforms the $\mathbb X$- and the $\mathbb {XZ}$-ansatz, as for QAOA a better approximation ratio can be obtained with fewer variational parameters.   
One may then wonder whether we can combine these favorable aspects of QAOA and the $\mathbb X$- and the $\mathbb {XZ}$-ans\"atze, e.g., by modifying QAOA such that for a few variational parameters high approximation ratios $\geq 0.95$ are obtained, while increasing the number of variational parameters continuously improves the approximation ratios, or conversely, whether it is possible to avoid the drop in the approximation ratios observed in Fig. \ref{Fig:Xcirc}(b) and Fig. \ref{Fig:NonCommute}(c) when the $k$-body depth is increased.

A natural modification of QAOA is to incorporate $n$ independent local $X$ rotations at each layer \cite{Hadfield2019}, such that the classical ansatz \eqref{eq:classical} is now contained in the ansatz \eqref{eq:QAOA}. This modification is well-motivated, due to the fact that even on non-interacting spin problems, conventional QAOA (i.e., using only global $X$ rotations) can fail to converge \cite{mcclean2020low}. However, we note that this does not automatically guarantee that the ground state is then reachable for a smaller number of variational parameters, as reachability also depends on the initial state $\ket{\phi}$. And indeed, from the olive-green curve in Fig. \ref{fig:QAOAcomp}, we see that a modification of QAOA to include local rotations while having $\ket{\phi}=\ket{\bm{+}}$ does not substantially change the convergence behavior, suggesting that a lack of reachability may affect the performance in this case. In comparison, if additionally the initial state is changed to $\ket{\phi}=\ket{\bm{0}}$ (red), then $M=n$ variational parameters do allow for retaining high approximation ratios $\geq 0.96$. Furthermore, increasing $M$ allows for increasing the approximation ratio to $\approx 0.99$ at $M\approx 45$. These results suggest that in addition to the incorporation of non-commutativity into an ansatz, the ability to guarantee reachability of the set of solution states with shallow circuits can enhance the performance of variational quantum algorithms. However, we do note that continuing to increase $M$ causes the approximation ratio to drop again, indicating that further work is needed to fully understand the tradeoffs in how these aspects of algorithm design impact performance. We also note that the numerical simulations in Fig. \ref{fig:QAOAcomp} suggest that the modified version of QAOA (red) and the $\mathbb{XZ}$-ansatz require $M\approx 2^{n+1}$ variational parameters in order to achieve approximation ratios asymptotically approaching $1$.

\section{Conclusions}
The successful implementation of VQAs rests on the ability to solve the underlying optimization problem, which in turn is determined by the structure of the corresponding optimization landscape. The landscape structure depends on the problem under consideration, and on the parameterized quantum circuit ansatz used to evaluate the associated objective function. In this work, we have focused on VQAs for solving the combinatorial optimization problem MaxCut, and studied the interplay between the graph type, ansatz, and critical point structure of the optimization landscape in this setting. In particular, we introduced a family of ans\"atze consisting of mutually commuting elements generated by $k$-body Pauli $X$ operators, which we termed $\mathbb X$-ans\"atze. We proved that for generic graphs, an ansatz from this family containing exponentially many variational parameters yields an optimization landscape that consists of global optima and saddle points only. Next, we considered examples for which an ansatz from this family with polynomially many variational parameters yields a landscape free from local optima, but for which MaxCut is also efficiently solvable classically. We then showed that for a given graph, a polynomially sized ansatz from the considered family exists if and only if there exists a purely classical algorithm that allows for solving MaxCut efficiently. As a consequence, we concluded that this ansatz family cannot offer a superpolynomial advantage over purely classical schemes.

We went on to numerically study whether introducing non-commutative ansatz elements could improve performance. For the $\mathbb X$-ansatz and its variants, we found that the addition of $k$-body terms did not improve the landscape structure in a manner that yields monotonically improving approximation ratios. However, for sufficiently high $k$-body depth, non-commutative versions of the $\mathbb X$-ansatz did perform better than their commutative counterparts. In total, we believe that the $\mathbb X$-ansatz family can be useful as a baseline, upon which new ideas for further performance improvement can be studied. We note that for relatively shallow circuit depth, the best performance overall was obtained through a modified version of QAOA, motivated by our prior findings, that includes independent local rotations.  

More research is needed to determine how quantum resources contribute to favorable landscape properties and if effective ans\"atze can be identified in a problem-dependent manner in the future. To this end, it is interesting to note that a ``classical'' ansatz, consisting of only local rotations, in combination with a first-order gradient algorithm, achieves remarkably high $\approx 0.95$ approximation ratios, and performs slightly better on average than the purely classical Goemans and Williamson MaxCut approximation algorithm for graphs with sufficiently high vertex degree. Due to the provable existence of local optima when utilizing the former classical ansatz, this suggests that the optimization landscape could be favorable in the sense that gradient algorithms converge to critical points with high associated approximation ratios.

In order to improve beyond these results, it would be desirable to steadily remove local optima from the landscape, while retaining only local optima with high corresponding approximation ratios. Furthermore, an equally important goal would be to widening the basin of attraction of global optima and reducing those of the remaining local optima \cite{mcclean2020low}. In addition, analyzing the curvature of the optimization landscape can reveal information about an algorithm's robustness to noise \cite{chakrabarti2007quantum}, suggesting that it could be desirable to design future ans\"atze to take this into account in order to enhance the quality of NISQ implementations.

Future studies may benefit from using the former ``classical'' ansatz as a starting point, from which additional ``quantum'' features are added, i.e., through the design of adaptive ans\"atze \cite{zhu2020adaptive,Grimsley2019_Adapt,Tang2020_Qubit, Zhang2020_Mutual,magann2021feedbackbased}, e.g. seeded from local rotations and then grown by incorporating gates that introduce non-commutativity and entanglement with the goal of continuously improving the landscape structure such that higher approximation ratios are attainable. Such procedures could additionally be guided by characterizing the geometry \cite{haug2021capacity} and entangling power \cite{PRXQuantum.1.020319} of the parameterized circuit.

The ideal outcome would be the development of VQAs that are \emph{scalable}, and have optimization landscapes with a favorable critical point structure. The most favorable landscapes would be free of local optima, and also free from certain types of saddle points. That is, while in this work we did not distinguish between classes of saddle points, the convergence of classical optimization routines can be significantly hindered by the presence of degenerate saddle points, as opposed to more favorable ``strict'' saddle points. Furthermore, the most favorable landscapes would also be free of barren plateaus  \cite{mcclean2018barren, PhysRevA.102.042207, JColes} (see also Appendix  \ref{App:Barren}), where the gradient becomes exponentially small, thereby impeding the progress of optimization algorithms. In fact, taken together, ans\"atze with polynomially many variational parameters lacking in barren plateaus, degenerate saddles, and local optima would allow a global optimum to be found efficiently using first-order gradient methods \cite{pmlr-v40-Ge15, levy2016power, pmlr-v70-jin17a}. However, for problems like MaxCut, we do not expect that such efficient solutions will be feasible in general, as we do not anticipate that VQAs will be able to solve \textsf{NP}-hard problems efficiently. Nevertheless, the achievement of any subset of these goals for scalable VQA ans\"atze would be highly desirable. Looking ahead, we hope that the strategies outlined above will allow for systematically assessing how to improve VQA performance across a range of applications.

\section*{Acknowledgements} 
We thank B. Hanin, J. McClean, M. McConnell, and M. Zhandry for valuable conversations and insights.
C.A. acknowledges support from the ARO (Grant No. W911NF-19-1-0382). A.M. acknowledges support from the U.S. Department of Energy, Office of Science, Office of Advanced Scientific Computing Research, Department of Energy Computational Science Graduate Fellowship under Award Number DE-FG02-97ER25308. H. R. acknowledges support from the ARO (Grant No. W911NF-19-1-0382) for the theoretical aspects of the research and from the DOE STTR (Grant No. DE-SC0020618) for the algorithmic implications of the research.

The authors are pleased to acknowledge that the work reported on in this paper was substantially performed using the Princeton Research Computing resources at Princeton University which is consortium of groups led by the Princeton Institute for Computational Science and Engineering (PICSciE) and Office of Information Technology's Research Computing.

This report was prepared as an account of work sponsored by an agency of the United States Government. Neither the United States Government nor any agency thereof, nor any of their employees, makes any warranty, express or implied, or assumes any legal liability or responsibility for the accuracy, completeness, or usefulness of any information, apparatus, product, or process disclosed, or represents that its use would not infringe privately owned rights. Reference herein to any specific commercial product, process, or service by trade name, trademark, manufacturer, or otherwise does not necessarily constitute or imply its endorsement, recommendation, or favoring by the United States Government or any agency thereof. The views and opinions of authors expressed herein do not necessarily state or reflect those of the United States Government or any agency thereof.

\newpage

\bibliography{VQA_landscape.bib}

\onecolumngrid

\appendix

\section{Proof of Lemma \eqref{theoremNonEigen}}\label{App:Lemma1Proof}
We first provide an outline of the proof, including major components and equations, in appendix \ref{subsec:appendix1Outline}, followed by a detailed proof in appendix \ref{subsec:appendix1Details}.

\subsection{\label{subsec:appendix1Outline} Proof Outline}
For ease and consistency of notation throughout the proof, we implicitly associate $H_{j}$ with the vertex subset $S_{j}$ on which $H_{j}$ acts.
The first observation is that for all $j$ such that $(a,b)\not\in \mathsf{Cut}(H_{j})$, we have that $e^{-i\theta_{j}H_{j}}$ commutes with $Z_{a}Z_{b}$. This observation motivates us to define the set $\mathcal{C}_{(a,b)} = \{H_{j}\in \mathcal{A}\mid (a,b)\in \mathsf{Cut}(H_{j})\}$ and $\mathcal{K}_{(a,b)}=\{K\subset \mathcal{C}_{(a,b)}\mid \oplus\{H\in K\}=\emptyset\}$, which allows for rewriting the objective function $J(\bm{\theta})$ as:
\begin{align}\label{eq:costSimplified1}
    J(\bm{\theta}) &= \sum_{(a,b)\in E} w_{a,b} \sum_{K\in \mathcal{K}_{(a,b)}} \Bigg[\prod_{H_{j}\in \mathcal{C}_{(a,b)}-K} \cos(2\theta_{j}) \prod_{H_{j}\in K} i\sin(2\theta_{j})\Bigg]
\end{align}
Notice that $\cos(2\theta_{k})$ and $\sin(2\theta_{k})$ never appear in the same product together, such that for a certain $k$ we can write:
\begin{align}\label{eq:costK}
    J(\bm{\theta}) &= \cos(2\theta_{k})S_{k} + \sin(2\theta_{k})T_{k} + V_{k},
\end{align}
where $S_{k},T_{k},V_{k}$ do not depend on $\theta_{k}$.
This result allows for easy expressions for the gradient and Hessian (where for ease of notation we use $\partial_{k} J(\bm{\theta})$ for the $k$-th gradient element and $\partial_{j,k}^{2}J(\bm{\theta})$ for the $j,k$ element of the Hessian):
\begin{align}\label{eq:gradHessSimp}
    \partial_{k} J(\bm{\theta}) &= 2\Big[-\sin(2\theta_{k})S_{k} + \cos(2\theta_{k})T_{k}\Big]\\
    \partial^{2}_{k,k} J(\bm{\theta}) &= -4\Big[\cos(2\theta_{k})S_{k} + \sin(2\theta_{k})T_{k}\Big]
\end{align}
The condition that $\partial_{k}J(\bm{\theta})=0$ at a critical point yields:
\begin{align}\label{eq:hessK}
    \sin(2\theta_{k})S_{k} &= \cos(2\theta_{k})T_{k}
\end{align}
Extensive analysis on the relative signs of $\sin(2\theta_{k}),\cos(2\theta_{k}),S_{k},T_{k}$ yields that nondegenerate critical points are either eigenstates of $H_{p}$ or saddle points, where we use the following definition of a saddle point:
\begin{definition}\label{def:saddle}
A parameter configuration $\bm{\theta}^{*}\in \mathbb{R}^{M}$ is a \underline{saddle} point of $J$ if it is a critical point, but for all $\epsilon>0$ there exist $\bm{\theta}_{1},\bm{\theta}_{2}$ with $\norm{\bm{\theta}^{*}-\bm{\theta}_{1}},\norm{\bm{\theta}^{*}-\bm{\theta}_{2}}<\epsilon$ and $J(\bm{\theta}_{1})<J(\bm{\theta}^{*})<J(\bm{\theta}_{2})$.
Furthermore, it is well-known \cite{multivariable2009} that a sufficient condition for $\bm{\theta}^{*}$ being a saddle point is that the set of scalars $\{z^{T}\Big(\nabla^{2}J(\bm{\theta}^{*})
\Big)z: z\in \mathbb{R}^{M}\}$ contains both positive and negative elements, where $\nabla^{2}J(\bm{\theta}^{*})$ is the Hessian matrix of $J$ evaluated at $\bm{\theta}^{*}$.
\end{definition}
Throughout the proof, we utilize either the former definition or the latter sufficient condition to show that a particular nondegenerate critical point is either a saddle point or an eigenstate of $H_{p}$. 

For the case of degenerate critical points, we justify and apply Proposition 3 from \cite{Bolis1980}, which states (reworded here from \cite{Bolis1980} to fit our context and notation):
\begin{proposition}\label{prop:bolisDegen}
Let $f(x)=\sum_{j=0}^{\infty} \frac{1}{j!}F_{j}(x)$ be the Taylor expansion of a function $f:\mathbb{R}^{M}\rightarrow \mathbb{R}$, where $F_{j}$ is the $j$-th Taylor form.
Let $F_{p}$ denote the first nonzero Taylor form of $f$ at a critical point $a$ of $f$, let $K_{p}$ denote the kernel of $F_{p}$, and let $F_{s}$ be the first Taylor form that does not vanish identically on $K_{p}$ (noting that at a critical point $2\leq p<s$, and $s$ may not exist).
Now, suppose that $F_{p}$ is positive semi-definite but not positive definite.
If $s$ exists and $F_{s}$ takes a negative value on $K_{p}$, or if $s$ does not exist, then $a$ is a saddle point of $f$.
\end{proposition}
Applying Proposition \eqref{prop:bolisDegen} yields that any degenerate critical point is a saddle, so that combined with the cases for nondegenerate critical points, we obtain that any parameter configuration at a critical point not corresponding to an eigenstate of $H_{p}$ is a saddle. \hfill $\square$

\subsection{\label{subsec:appendix1Details} Details of the Proof}
\subsubsection{Computation of the objective function} 

We first remark that much of the computations of this section are only to derive the form of $J(\bm{\theta})$ as in equation \eqref{eq:costDep}, which has also been shown in similar forms with respect to the well-known parameter shift rule \cite{Meyer2021}.
Readers interested in the main portion of the proof can thus skip straight to equation \eqref{eq:costDep}.

To derive a manipulable form of the cost function, we first compute $J(\bm{\theta})$, which we expand by utilizing linearity, to yield a form that is easier to manipulate:
\begin{align}
    J(\bm{\theta}) &= \bra{\varphi(\bm{\theta})}H_{p}\ket{\varphi(\bm{\theta})} \nonumber\\
    &= \sum_{(a,b)\in E} w_{a,b}\bra{\varphi(\bm{\theta})}Z_{a}Z_{b}{\ket{\varphi(\bm{\theta})}} \nonumber\\
    &= \sum_{(a,b)\in E} w_{a,b}\bra{\bm{0}}\Big[(\prod_{j=1}^{M} e^{i\theta_{j}H_{j}})Z_{a}Z_{b}(\prod_{j=1}^{M} e^{-i\theta_{j}H_{j}})\Big]\ket{\bm{0}}
\end{align}
Here, notice that for all $j$ such that $(a,b)\not\in \mathsf{Cut}(H_{j})$, we have that $e^{i\theta_{j}H_{j}}$ commutes with $Z_{a}Z_{b}$.
For ease of notation, let $\mathcal{C}_{(a,b)}=\{H_{j}\in \mathcal{A}\mid (a,b)\in \mathsf{Cut}(H_{j})\}$, corresponding to the elements that do not commute with $Z_{a}Z_{b}$.
Thus:
\begin{align}
    J(\bm{\theta}) &= \sum_{(a,b)\in E} w_{a,b} \bra{\bm{0}}\Big[(\prod_{H_{j}\in \mathcal{C}_{(a,b)}} e^{i\theta_{j}H_{j}})Z_{a}Z_{b}(\prod_{H_{j}\in \mathcal{C}_{(a,b)}} e^{-i\theta_{j}H_{j}})\Big]\ket{\bm{0}}
\end{align}
Now, we first recall the well-known formula for matrices $A,B$ with $B^{2}=\mathds{1}$:
\begin{align}
    e^{i\alpha B}Ae^{-i\alpha B} &= \cos^{2}(\alpha)A + \sin^{2}(\alpha)BAB + i\sin(\alpha)\cos(\alpha)[B,A]
\end{align}
Setting $A=Z_{a}Z_{b}$, $B=H_{j}$, and $\alpha=\theta_{j}$ yields:
\begin{align}
    e^{i\theta_{j} H_{j}}Z_{a}Z_{b}e^{-i\theta_{j} H_{j}} &= \cos^{2}(\theta_{j}) Z_{a}Z_{b} + \sin^{2}(\theta_{j})H_{j}Z_{a}Z_{b}H_{j} + i\sin(\theta_{j})\cos(\theta_{j})[H_{j},Z_{a}Z_{b}]
\end{align}
For $H_{j}\in \mathcal{C}_{(a,b)}$, without loss of generality assume $\{a,b\}\cap H_{j}=\{a\}$.
Then, using simple algebra we obtain:
\begin{align}
    e^{i\theta_{j} H_{j}}Z_{a}Z_{b}e^{-i\theta_{j} H_{j}} &= \cos^{2}(\theta_{j}) Z_{a}Z_{b} + \sin^{2}(\theta_{j})H_{j}Z_{a}Z_{b}H_{j} + i\sin(\theta_{j})\cos(\theta_{j})[H_{j},Z_{a}Z_{b}] \nonumber\\
    &= \cos^{2}(\theta_{j})Z_{a}Z_{b} + \sin^{2}(\theta_{j})X_{a}Z_{a}Z_{b}X_{a} + i\sin(\theta_{j})\cos(\theta_{j})(H_{j}X_{a}Z_{b})[X_{a},Z_{a}] \nonumber\\
    &= \cos(2\theta_{j})Z_{a}Z_{b} + \sin(2\theta_{j})H_{j}X_{a}Z_{b}Y_{a} \nonumber\\
    &= \cos(2\theta_{j})Z_{a}Z_{b} + i\sin(2\theta_{j})H_{j}Z_{a}Z_{b} \nonumber\\
    &= (\cos(2\theta_{j})\mathds{1} + i\sin(2\theta_{j})H_{j})Z_{a}Z_{b}
\end{align}
Thus, since each of the $H_{j}$'s commute and $Z_{a}Z_{b}\ket{\bm{0}}=\ket{\bm{0}}$ for all $(a,b)\in E$, we arrive at: 
\begin{align}
    J(\bm{\theta}) = \sum_{(a,b)\in E} w_{a,b}\bra{\bm{0}}\Big[\prod_{H_{j}\in \mathcal{C}_{(a,b)}} \Big(\cos(2\theta_{j}) \mathds{1} + i\sin(2\theta_{j}) H_{j}\Big)\Big]\ket{\bm{0}}
\end{align}
Now, notice that the summands in the expansion of $\prod_{H_{j}\in \mathcal{C}_{(a,b)}}$ are products over all $H_{j}\in \mathcal{C}_{(a,b)}$ of either $\cos(2\theta_{j}) \mathds{1}$ or $i\sin(2\theta_{j}) H_{j}$.
Taking the expectation of a particular summand with respect to the state $\ket{\bm{0}}$ yields a non-zero value if and only if the product of the $H_{j}$'s that are included is the identity; since each $H_{j}$ corresponds to a particular vertex subset, this condition can also be expressed as the symmetric difference of all included $H_{j}$'s being the empty set.
This leads us to define $\mathcal{K}_{(a,b)}=\{K\subset \mathcal{C}_{(a,b)}\mid \oplus\{H\in K\}=\emptyset\}$, where for a particular $K\in \mathcal{K}_{(a,b)}$ the elements $H_{j}\in K$ are those corresponding to the $i\sin(2\theta_{j})$ terms, and the elements $H_{j}\in \mathcal{C}_{(a,b)}-K$ correspond to the $\cos(2\theta_{j})$ terms.
This yields:
\begin{align}
\label{eq:appcost}
    J(\bm{\theta}) &= \sum_{(a,b)\in E} w_{a,b} \sum_{K\in \mathcal{K}_{(a,b)}} \Bigg[\prod_{H_{j}\in \mathcal{C}_{(a,b)}-K} \cos(2\theta_{j})\prod_{H_{j}\in K} i\sin(2\theta_{j})\Bigg]
\end{align}
While the presence of the imaginary unit $i$ in the product may seem problematic given that the expectation $J$ must be real here, we can see that $\abs{K}$ is even for all $K\in \mathcal{K}_{(a,b)}$ (since in the construction of $\mathcal{K}_{(a,b)}$ we note that elements in $K$ must have exactly one of either $a$ or $b$, each of which requires an even number of elements in order to attain an empty symmetric difference), so that the product remains real.

Now, consider the gradient element $\partial_{k}J(\bm{\theta})\equiv \frac{\partial}{\partial \theta_{k}}J(\bm{\theta})$ corresponding to an ansatz element defined by some $H_{k}$:
\begin{align}
    \partial_{k} J(\bm{\theta}) &= \frac{\partial}{\partial \theta_{k}} \sum_{(a,b)\in E} w_{a,b} \sum_{K\in \mathcal{K}_{(a,b)}} \Bigg[\prod_{H_{j}\in \mathcal{C}_{(a,b)}-K} \cos(2\theta_{j})\prod_{H_{j}\in K} i\sin(2\theta_{j})\Bigg]
\end{align}
Since $\theta_{k}$ appears only if $H_{k}\in \mathcal{C}_{(a,b)}$, only edges $(a,b)$ in $\mathsf{Cut}(H_{k})$ have a non-zero partial derivative:
\begin{align}
    \partial_{k} J(\bm{\theta}) &= \sum_{(a,b)\in \mathsf{Cut}(H_{k})} w_{a,b}\sum_{K\in \mathcal{K}_{(a,b)}} \frac{\partial}{\partial \theta_{k}}\Bigg[\prod_{H_{j}\in \mathcal{C}_{(a,b)}-K} \cos(2\theta_{j})\prod_{H_{j}\in K} i\sin(2\theta_{j})\Bigg]
\end{align}
Since $\theta_{k}$ appears exactly once in the product $\prod_{H_{j}\in \mathcal{C}_{(a,b)}-K} \cos(2\theta_{j})\prod_{H_{j}\in K} i\sin(2\theta_{j})$ (either as $\cos(2\theta_{k})$ or $i\sin(2\theta_{k})$ according to whether $H_{k}\in K$ or not), we can split this as follows:
\begin{align}
    \partial_{k} J(\bm{\theta}) &= \sum_{(a,b)\in \mathsf{Cut}(H_{k})} w_{a,b} \Bigg(\frac{\partial}{\partial \theta_{k}} \sum_{K\in \mathcal{K}_{(a,b)}~s.t.~H_{k}\not\in K} \Bigg[\prod_{H_{j}\in \mathcal{C}_{(a,b)}-K} \cos(2\theta_{j})\prod_{H_{j}\in K} i\sin(2\theta_{j})\Bigg] \nonumber\\
    &~~~~~~~~~~~~~~~~~~~~~~~ + \frac{\partial}{\partial \theta_{k}} \sum_{K\in \mathcal{K}_{(a,b)}~s.t.~H_{k}\in K} \Bigg[\prod_{H_{j}\in \mathcal{C}_{(a,b)}-K} \cos(2\theta_{j})\prod_{H_{j}\in K} i\sin(2\theta_{j})\Bigg] \Bigg)
\end{align}
Now, notice that in order to replace $\cos(2\theta_{k})$ with its derivative $-2\sin(2\theta_{k})$, we can multiply by $\frac{-2\sin(2\theta_{k})}{\cos(2\theta_{k})}$, and similarly for the derivative of $\sin(2\theta_{k})$:
\begin{align}
    \partial_{k}J(\bm{\theta}) &= \sum_{(a,b)\in \mathsf{Cut}(H_{k})} w_{a,b} \Bigg(\frac{-2\sin(2\theta_{k})}{\cos(2\theta_{k})} \sum_{K\in \mathcal{K}_{(a,b)}~s.t.~H_{k}\not\in K} \Bigg[\prod_{H_{j}\in \mathcal{C}_{(a,b)}-K} \cos(2\theta_{j})\prod_{H_{j}\in K} i\sin(2\theta_{j})\Bigg] \nonumber\\
    &~~~~~~~~~~~~~~~~~~~~~~~~~ + \frac{2\cos(2\theta_{k})}{\sin(2\theta_{k})} \sum_{K\in \mathcal{K}_{(a,b)}~s.t.~H_{k}\in K} \Bigg[\prod_{H_{j}\in \mathcal{C}_{(a,b)}-K} \cos(2\theta_{j})\prod_{H_{j}\in K} i\sin(2\theta_{j})\Bigg] \Bigg) \nonumber\\
    &= \frac{-2\sin(2\theta_{k})}{\cos(2\theta_{k})}\Bigg(\sum_{(a,b)\in \mathsf{Cut}(H_{k})} w_{a,b} \sum_{K\in \mathcal{K}_{(a,b)}~s.t.~H_{k}\not\in K} \Bigg[\prod_{H_{j}\in \mathcal{C}_{(a,b)}-K} \cos(2\theta_{j})\prod_{H_{j}\in K} i\sin(2\theta_{j})\Bigg]\Bigg) \nonumber\\
    &~~~~~~ + \frac{2\cos(2\theta_{k})}{\sin(2\theta_{k})}\Bigg(\sum_{(a,b)\in \mathsf{Cut}(H_{k})} w_{a,b}\sum_{K\in \mathcal{K}_{(a,b)}~s.t.~H_{k}\in K} \Bigg[\prod_{H_{j}\in \mathcal{C}_{(a,b)}-K} \cos(2\theta_{j})\prod_{H_{j}\in K} i\sin(2\theta_{j})\Bigg]\Bigg)
\end{align}
Now, for ease of notation in arguing the remainder of this proof, define:
\begin{align}
    \label{eq:defSk} S_{k} &= \frac{1}{\cos(2\theta_{k})}\sum_{(a,b)\in \mathsf{Cut}(H_{k})} w_{a,b} \sum_{K\in \mathcal{K}_{(a,b)}~s.t.~H_{k}\not\in K} \Bigg[\prod_{H_{j}\in \mathcal{C}_{(a,b)}-K} \cos(2\theta_{j})\prod_{H_{j}\in K} i\sin(2\theta_{j})\Bigg]\\
    \label{eq:defTk} T_{k} &= \frac{1}{\sin(2\theta_{k})}\sum_{(a,b)\in \mathsf{Cut}(H_{k})} w_{a,b}\sum_{K\in \mathcal{K}_{(a,b)}~s.t.~H_{k}\in K} \Bigg[\prod_{H_{j}\in \mathcal{C}_{(a,b)}-K} \cos(2\theta_{j})\prod_{H_{j}\in K} i\sin(2\theta_{j})\Bigg]
\end{align}
Notice in particular that $S_{k},T_{k}$ do not depend on $\theta_{k}$, by construction since the prefactor $\frac{1}{\cos(2\theta_{k})}$ cancels the existing $\cos(2\theta_{k})$ term in the product for $S_{k}$, and similarly for $T_{k}$.
Thus, we have:
\begin{align}
    \partial_{k} J(\bm{\theta}) &= -2\sin(2\theta_{k})S_{k} + 2\cos(2\theta_{k})T_{k} = 2\Big[-\sin(2\theta_{k})S_{k} + \cos(2\theta_{k})T_{k}\Big]
\end{align}
We can then easily compute the diagonal Hessian elements $\partial^{2}_{k,k} J(\bm{\theta})\equiv\frac{\partial^{2}}{\partial\theta_{k}^{2}}J(\bm{\theta})$ as well:
\begin{align}
    \partial^{2}_{k,k} J(\bm{\theta}) &= 2\Big[-2\cos(2\theta_{k})S_{k} - 2\sin(2\theta_{k})T_{k}\Big] = -4\Big[\cos(2\theta_{k})S_{k} + \sin(2\theta_{k})T_{k}\Big]
\end{align}
Using this result, we can expand $S_{k},T_{k}$ in the formula for the diagonal Hessian element to relate its value to the objective function $J(\bm{\theta})$:
\begin{align}
    -\frac{\partial^{2}_{k,k} J(\bm{\theta})}{4} &= \cos(2\theta_{k})S_{k} + \sin(2\theta_{k})T_{k} \nonumber\\
    &= \cos(2\theta_{k})\cdot \frac{1}{\cos(2\theta_{k})}\sum_{(a,b)\in \mathsf{Cut}(H_{k})} w_{a,b} \sum_{K\in \mathcal{K}_{(a,b)}~s.t.~H_{k}\not\in K} \Bigg[\prod_{H_{j}\in \mathcal{C}_{(a,b)}-K} \cos(2\theta_{j})\prod_{H_{j}\in K} i\sin(2\theta_{j})\Bigg] \nonumber\\
    &~~~~~~ + \sin(2\theta_{k})\cdot \frac{1}{\sin(2\theta_{k})}\sum_{(a,b)\in \mathsf{Cut}(H_{k})} w_{a,b}\sum_{K\in \mathcal{K}_{(a,b)}~s.t.~H_{k}\in K} \Bigg[\prod_{H_{j}\in \mathcal{C}_{(a,b)}-K} \cos(2\theta_{j})\prod_{H_{j}\in K} i\sin(2\theta_{j})\Bigg] \nonumber\\
    &= \sum_{(a,b)\in \mathsf{Cut}(H_{k})} w_{a,b} \sum_{K\in \mathcal{K}_{(a,b)}~s.t.~H_{k}\not\in K} \Bigg[\prod_{H_{j}\in \mathcal{C}_{(a,b)}-K} \cos(2\theta_{j})\prod_{H_{j}\in K} i\sin(2\theta_{j})\Bigg] \nonumber\\
    &~~~~~~ + \sum_{(a,b)\in \mathsf{Cut}(H_{k})} w_{a,b}\sum_{K\in \mathcal{K}_{(a,b)}~s.t.~H_{k}\in K} \Bigg[\prod_{H_{j}\in \mathcal{C}_{(a,b)}-K} \cos(2\theta_{j})\prod_{H_{j}\in K} i\sin(2\theta_{j})\Bigg]
\end{align}
Since we sum over all $K\in \mathcal{K}_{(a,b)}$ such that $H_{k}\not\in K$ and those such that $H_{k}\in K$, this is in fact a complete sum:
\begin{align}
    -\frac{\partial_{k,k}^{2}J(\bm{\theta})}{4} &= \sum_{(a,b)\in \mathsf{Cut}(H_{k})} w_{a,b} \sum_{K\in \mathcal{K}_{(a,b)}} \Bigg[\prod_{H_{j}\in \mathcal{C}_{(a,b)}-K} \cos(2\theta_{j})\prod_{H_{j}\in K} i\sin(2\theta_{j})\Bigg] \nonumber\\
    &= \sum_{(a,b)\in E} w_{a,b} \sum_{K\in \mathcal{K}_{(a,b)}} \Bigg[\prod_{H_{j}\in \mathcal{C}_{(a,b)}-K} \cos(2\theta_{j})\prod_{H_{j}\in K} i\sin(2\theta_{j})\Bigg] \nonumber\\
    &~~~~~~- \sum_{(a,b)\not\in \mathsf{Cut}(H_{k})} w_{a,b} \sum_{K\in \mathcal{K}_{(a,b)}} \Bigg[\prod_{H_{j}\in \mathcal{C}_{(a,b)}-K} \cos(2\theta_{j})\prod_{H_{j}\in K} i\sin(2\theta_{j})\Bigg] \nonumber\\
    &= J(\bm{\theta}) - \sum_{(a,b)\not\in \mathsf{Cut}(H_{k})} w_{a,b} \sum_{K\in \mathcal{K}_{(a,b)}} \Bigg[\prod_{H_{j}\in \mathcal{C}_{(a,b)}-K} \cos(2\theta_{j})\prod_{H_{j}\in K} i\sin(2\theta_{j})\Bigg]
\end{align}
Thus, we have:
\begin{align}
    J(\bm{\theta}) &= -\frac{\partial^{2}_{k,k} J(\bm{\theta})}{4} + \sum_{(a,b)\not\in \mathsf{Cut}(H_{k})} w_{a,b} \sum_{K\in \mathcal{K}_{(a,b)}} \Bigg[\prod_{H_{j}\in \mathcal{C}_{(a,b)}-K} \cos(2\theta_{j})\prod_{H_{j}\in K} i\sin(2\theta_{j})\Bigg]
\end{align}
Here, since as mentioned above $\theta_{k}$ appears only if $H_{k}\in \mathcal{C}_{(a,b)}$, the second summand does not depend on $\theta_{k}$; this allows us to proceed with our analysis easily.

We now can summarize the main ingredients required here: 
\begin{align}
    \label{eq:costDep} J(\bm{\theta}) &= \cos(2\theta_{k})S_{k} + \sin(2\theta_{k})T_{k} + (\text{term not dependent on $\theta_{k}$})\\
  \label{eq:gradientk}  \partial_{k} J(\bm{\theta}) &= 2\Big[-\sin(2\theta_{k})S_{k} + \cos(2\theta_{k})T_{k}\Big]\\
   \label{eq:hessiank} \partial^{2}_{k,k} J(\bm{\theta}) &= -4\Big[\cos(2\theta_{k})S_{k} + \sin(2\theta_{k})T_{k}\Big]
\end{align}
At a critical point $\bm{\theta}^{*}$, we have that for each $k$ that $\partial_{k}J(\bm{\theta})=0$. From \eqref{eq:gradientk} we thus have that at a critical point the condition 
\begin{align}\label{eq:simpleCritCond}
    \sin(2\theta_{k})S_{k} &= \cos(2\theta_{k})T_{k}
\end{align}
has to hold for all $k$. 

\subsubsection{Case considerations} 
In order to prove that any critical point $\bm{\theta}^{*}$ not corresponding to an eigenstate of $H_{p}$ is a saddle point, we consider the cases (a)-(c) below. 
We first consider case \textbf{(a)}, and show that if for some $k$ both sides of \eqref{eq:simpleCritCond} are equal but non-zero, meaning that $\sin(2\theta_{k})\neq 0$, $\cos(2\theta_{k})\neq 0$, $S_{k}\neq 0$, and $T_{k}\neq 0$, the corresponding critical points correspond to saddle points.
We then go on to consider the case \textbf{(b)} where for some $k$ both sides are zero with $S_{k}=T_{k}=0$, showing that in this case the corresponding critical points must be saddle points too. 
Here we will distinguish between (i) non-degenerate and (ii) degenerate critical points, where for case (ii) we utilize Proposition 1. 
As the cases (a) and (b) correspond to saddle points, the only case left for which a critical point $\bm{\theta}^{*}$ could potentially not correspond to a saddle point is \textbf{(c)}, where for each $k$ either $\cos(2\theta_{k})=S_{k}=0$ or $\sin(2\theta_{k})=T_{k}=0$, so that together condition \eqref{eq:simpleCritCond} is satisfied for all $k$. 
We finally show that such critical points satisfying (c) correspond to eigenstates of $H_{p}$. \\

\noindent 
\textbf{Case (a):} There exists $k$ such that $\sin(2\theta_{k})\neq 0$, $\cos(2\theta_{k})\neq 0$, $S_{k}\neq 0$, $T_{k}\neq 0$\\
\\
If for some $k$ the above holds, we can rewrite condition \eqref{eq:simpleCritCond} as $\frac{\sin(2\theta_{k})}{\cos(2\theta_{k})}=\frac{T_{k}}{S_{k}}$. The objective function $J(\bm{\theta})$ given by \eqref{eq:costDep} can then be rewritten as: 
\begin{align}
    J(\bm{\theta}) &= \cos(2\theta_{k})\cdot \frac{S_{k}^{2}+T_{k}^{2}}{S_{k}} + (\text{term not dependent on $\theta_{k}$})
\end{align}
We note again that $S_{k}$ and $T_{k}$ are independent of $\theta_{k}$.
Moreover, since $\sin(2\theta_{k})\neq 0$ we have $\cos(2\theta_{k})\neq \pm 1$. 
As such, for all $\epsilon>0$ the $\epsilon$-ball around $\theta_{k}$ both increases and decreases the value of $\cos\Big(2(\theta_{k}\pm \epsilon)\Big)$. 
Thus, by definition \eqref{def:saddle} all critical points corresponding to case (a) are saddle points.\\

\noindent 
\textbf{Case (b):} There exists $k$ such that $S_{k}=T_{k}=0$ 

\begin{itemize}
	\item [(i)] Non-degenerate case (invertible Hessian)\\
	 We first note that if $S_{k}=T_{k}=0$ for some $k$, from \eqref{eq:hessiank} we see that then $\partial_{k,k}^{2}J(\bm{\theta})=0$. 
	 Now, consider the vector $u$ with 1 in the $k$-th entry and $0$ everywhere else. 
	 Since the Hessian matrix $\nabla^{2}J(\bm{\theta})$ is invertible by definition, so that $\Big(\nabla^{2}J(\bm{\theta})\Big)u\neq \bm{0}$, there exists some $l$ such that the off-diagonal Hessian element $\partial_{k,l}^{2}J(\bm{\theta})\neq 0$. 
	 For $t\in \mathbb{R}$, consider the set of vectors $v_{t}$ with value $t$ in the $l$-th position, $1$ in the $k$-th position, and $0$ everywhere else, so that 
	 \begin{align} \label{eq:app1}
		v_{t}^{T}\Big(\nabla^{2}J(\bm{\theta})\Big)v_{t}=t^{2}\partial_{l,l}^{2}J(\bm{\theta}) + 2t\partial_{k,l}^{2}J(\bm{\theta}).
	\end{align} 
	If $\partial_{l,l}^{2}J(\bm{\theta})=0$, then \eqref{eq:app1} is linear in $t$. 
	Thus, in this case the left-hand side attains for $t\in \mathbb{R}$ both positive and negative values. 
	If $\partial_{l,l}^{2}J(\bm{\theta})\neq 0$, then \eqref{eq:app1} is quadratic in $t$ with roots at $t=0$ and $t=-\frac{2\partial_{k,l}^{2}J(\bm{\theta})}{\partial_{l,l}^{2}J(\bm{\theta})}\neq 0$. 
	As such, here the left-hand side attains for $t\in \mathbb{R}$ both positive and negative values too. 
	From the sufficient condition in the definition \eqref{def:saddle} of a saddle point we conclude that in the case (b)(i) the corresponding critical points correspond to saddle points.

\item[(ii)] Degenerate case (non-invertible Hessian)\\
We first show that in order to obtain a critical point not corresponding to a saddle point, the Hessian $\nabla^{2}J(\bm{\theta})$ at these points must be diagonal. 
As above, since $S_{k}=T_{k}=0$ we know that $\partial_{k,k}^{2}J(\bm{\theta})=0$.
If the Hessian is not diagonal, then there exists some $l$ such that the off-diagonal Hessian element $\partial_{k,l}^{2}J(\bm{\theta})\neq 0$, which allows for proceeding as in case (b)(i). 
We conclude that critical points yielding a non-diagonal Hessian correspond to saddle points.  
 
Therefore, at critical points that do not immediately correspond to saddle points the Hessian must be diagonal with at least one element being $0$. 
We treat this case by applying Proposition \eqref{prop:bolisDegen} from above, identifying the function $f(x)$ as the objective function $J(\bm{\theta})$. Without loss of generality, assume the Hessian is positive semi-definite (as an analogous argument holds for the negative semi-definite case). 
Notice that the kernel of the Hessian (which we denote by $K_{p}$ as in the setting of Proposition \eqref{prop:bolisDegen}) is the set of vectors with zeros in all indices $j$ for which the $j$-th diagonal Hessian element $\partial_{j,j}^{2}J(\bm{\theta})$ is non-zero, and any real values for elements corresponding to other indices.
Furthermore, as in the setting of Proposition \eqref{prop:bolisDegen} consider the case where $s$ exists (since otherwise we immediately obtain that we have a saddle) such that $F_{s}$ is the first Taylor form that does not vanish identically on $K_{p}$.
Without loss of generality, assume $s=3$ (as an analogous argument holds for $s>3$), and let $\mathfrak{T}$ be the order-3 tensor representing the third-derivatives of $J$.
First, notice that the ``diagonal" elements satisfy $\partial_{k,k,k}^{3}J(\bm{\theta})=-8\Big[-\sin(2\theta_{k})S_{k}+\cos(2\theta_{k})T_{k}\Big]=-4\partial_{k}J(\bm{\theta})=0$, which implies that if the kernel $K_{p}$ is one-dimensional, $\mathfrak{T}$ is zero identically on $K_{p}$.
As such, in order for $\mathfrak{T}$ not to vanish identically on $K_{p}$, there must be some set of indices $j,k,l$ such that $\partial_{j,j}^{2}J(\bm{\theta})=\partial_{k,k}^{2}J(\bm{\theta})=\partial_{l,l}^{2}J(\bm{\theta})=0$ but $\partial_{j,k,l}^{3}J(\bm{\theta})\neq 0$.
Now, let $u\in K_{p}$ be the vector with 1 in the $j$-th, $k$-th, and $l$-th entries and 0 everywhere else.
Notice that by construction, applying $\mathfrak{T}$ to $u$ and $-u$ yields non-zero values with opposite signs, so in particular $\mathfrak{T}$ takes a negative value on $K_{p}$.
We can now apply Proposition \eqref{prop:bolisDegen}, which states that we have a saddle point in this case.
\end{itemize}

\noindent 
\textbf{Case (c)} For all $k$, either $\sin(2\theta_{k})=0$ and $T_{k}=0$, or $\cos(2\theta_{k})=0$ and $S_{k}=0$\\
\\
From \eqref{eq:appcost} we recall the form of the objective function: 
\begin{align}
    J(\bm{\theta}) &= \sum_{(a,b)\in E} w_{a,b} \sum_{K\in \mathcal{K}_{(a,b)}} \Bigg[\prod_{H_{j}\in \mathcal{C}_{(a,b)}-K} \cos(2\theta_{j})\prod_{H_{j}\in K} i\sin(2\theta_{j})\Bigg]
\end{align}
Within the set $\mathcal{C}_{(a,b)}$, each $H_{j}$ either has $\sin(2\theta_{j})=0$ or $\cos(2\theta_{j})=0$.
This means that at most one of the products can be non-zero, since all other products will contain at least one configuration with $\sin(2\theta_{j})$ or $\cos(2\theta_{j})$ being 0.
This non-zero product therefore either takes the value $1$ or $-1$ (as it is a product of $\sin$ and $\cos$ that are each either 1 or $-1$), yielding:
\begin{align}
    J(\bm{\theta}) &= \sum_{(a,b)\in E} (\pm w_{a,b}),
\end{align}
where the sign of a particular $w_{a,b}$ is determined by whether the number of $H_{j}\in \mathcal{C}_{(a,b)}$ that have $\cos(2\theta_{j})=-1$ or $\sin(2\theta_{j})=-1$ is even or odd.
As discussed in section \ref{sec:ansatzBackground}, this assignment of signs precisely corresponds to a cut of the graph, which in turn corresponds to an eigenstate of $H_{p}$, as desired.

At this juncture we also remark that an alternate proof of case (c) can be seen by noticing that $\sin(2\theta_{k})=0$ or $\cos(2\theta_{k})=0$, implying $\theta_{k}=n\pi$ and $\theta_k =(2n+1)/2 \pi $ respectively, such that all gates belong to full flips of qubits or a global phase.
As the circuit starts in an eigenstate (namely, $\ket{\bm{0}}$), the circuit also ends in an eigenstate, as desired.

 \hfill $\square$

\section{ Barren Plateaus}\label{App:Barren}

The phenomenon of barren plateaus has recently been considered as one of the main bottlenecks for VQAs \cite{VQAReview}. A barren plateau appears if the gradient of the objective function becomes exponentially small, which is typically analyzed by randomly ``sampling'' the optimization landscape given by $J(\bm{\theta})$. That is, a barren plateau appears if the variance $\text{Var}\Big(\partial_{k}J(\bm{\theta})\Big)$ of the components of the gradient becomes exponentially small in the number of qubits $n$, while the expectation $\mathbb{E}[\partial_{k}J(\bm{\theta})]$ vanishes for all $k$. Here, we compute the variances explicitly for the $\mathbb{X}$-ansatz. The expectations are taken over the parameters $\bm{\theta}$, each of which are considered to be independently and identically uniformly distributed in $[0,2\pi)$. Throughout the computations, we utilize the well-known identities $\mathbb{E}[\sin(x)]=\mathbb{E}[\cos(x)]=0$ and $\mathbb{E}[\sin^{2}(x)]=\mathbb{E}[\cos^{2}(x)]=\frac{1}{2}$. 

In order to compute the variances, we first recall from \eqref{eq:gradientk} the form of the gradient element $\partial_{k}J(\bm{\theta})$:
\begin{align}
    \partial_{k} J(\bm{\theta}) &= \frac{-2\sin(2\theta_{k})}{\cos(2\theta_{k})}\Bigg(\sum_{(a,b)\in \mathsf{Cut}(H_{k})} w_{a,b} \sum_{K\in \mathcal{K}_{(a,b)}~s.t.~H_{k}\not\in K} \Bigg[\prod_{H_{j}\in \mathcal{C}_{(a,b)}-K} \cos(2\theta_{j})\prod_{H_{j}\in K} i\sin(2\theta_{j})\Bigg]\Bigg)\nonumber \\
    &~~~~~~ + \frac{2\cos(2\theta_{k})}{\sin(2\theta_{k})}\Bigg(\sum_{(a,b)\in \mathsf{Cut}(H_{k})} w_{a,b}\sum_{K\in \mathcal{K}_{(a,b)}~s.t.~H_{k}\in K} \Bigg[\prod_{H_{j}\in \mathcal{C}_{(a,b)}-K} \cos(2\theta_{j})\prod_{H_{j}\in K} i\sin(2\theta_{j})\Bigg]\Bigg)\\
    &= 2\Big[-\sin(2\theta_{k})S_{k} + \cos(2\theta_{k})T_{k}\Big]
\end{align}
Thus, since $S_{k},T_{k}$ are independent from $\theta_{k}$ and each have finite expectation, we can write:
\begin{align}
    \mathbb{E}[\partial_{k}J(\bm{\theta})] &= 2\Big[-\mathbb{E}[\sin(2\theta_{k})]\mathbb{E}[S_{k}] + \mathbb{E}[\cos(2\theta_{k})]\mathbb{E}[T_{k}]\Big]= 2\cdot (0\cdot \mathbb{E}[S_{k}] + 0\cdot \mathbb{E}[T_{k}]) = 0,
\end{align}
where we used the identities described above.
For the variance, we can first write:
\begin{align}
    \mathbb{E}[(\partial_{k} J(\bm{\theta}))^{2}] &= \mathbb{E}[4(\sin^{2}(2\theta_{k})S_{k}^{2} - 2\sin(2\theta_{k})\cos(2\theta_{k})S_{k}T_{k} + \cos^{2}(2\theta_{k})T_{k}^{2})] \nonumber\\
    &= 4\cdot \Big[\mathbb{E}[\sin^{2}(2\theta_{k})]\mathbb{E}[S_{k}^{2}] - \mathbb{E}[\sin(4\theta_{k})]\mathbb{E}[S_{k}T_{k}] + \mathbb{E}[\cos^{2}(2\theta_{k})]\mathbb{E}[T_{k}^{2}]\Big]\nonumber\\
    &= 4\cdot \frac{1}{2}\cdot \mathbb{E}[S_{k}^{2}] - 0 + 4\cdot \frac{1}{2}\cdot \mathbb{E}[T_{k}^{2}]\nonumber\\
    &= 2\cdot \mathbb{E}[S_{k}^{2}+T_{k}^{2}]
\end{align}
where we used linearity and the identities described above.
Since in the expansion of $S_{k}^{2},T_{k}^{2}$ the expectation of any single sine or cosine is zero, only the square of the terms survive.
This yields:
\begin{align}
    \mathbb{E}[(\partial_{k} J(\bm{\theta}))^{2}] &= \mathbb{E}\Bigg[\frac{2}{\cos^{2}(2\theta_{k})}\sum_{(a,b)\in \mathsf{Cut}(H_{k})} w_{a,b}^{2} \sum_{K\in \mathcal{K}_{(a,b)}~s.t.~H_{k}\not\in K} \Bigg[\prod_{H_{j}\in \mathcal{C}_{(a,b)}-K} \cos^{2}(2\theta_{j})\prod_{H_{j}\in K} -\sin^{2}(2\theta_{j})\Bigg]\Bigg]\nonumber\\
    \label{eq:appVar1} &~~~ + \mathbb{E}\Bigg[\frac{2}{\sin^{2}(2\theta_{k})}\sum_{(a,b)\in \mathsf{Cut}(H_{k})} w_{a,b}^{2} \sum_{K\in \mathcal{K}_{(a,b)}~s.t.~H_{k}\in K} \Bigg[\prod_{H_{j}\in \mathcal{C}_{(a,b)}-K} \cos^{2}(2\theta_{j})\prod_{H_{j}\in K} -\sin^{2}(2\theta_{j})\Bigg]\Bigg] \\
    \label{eq:appVar2} &= 2\cdot \sum_{(a,b)\in \mathsf{Cut}(H_{k})}w_{a,b}^{2}\sum_{K\in \mathcal{K}_{(a,b)}} \Big(\frac{1}{2}\Big)^{\abs{\mathcal{C}_{(a,b)}}-1}\\
    \label{eq:appVar3} &= 4\cdot \sum_{(a,b)\in \mathsf{Cut}(H_{k})}w_{a,b}^{2}\cdot \frac{\abs{\mathcal{K}_{(a,b)}}}{2^{\abs{\mathcal{C}_{(a,b)}}}}
\end{align}
where from \eqref{eq:appVar1} to \eqref{eq:appVar2} we used the identities $\mathbb{E}[\sin^{2}(2\theta_{j})]=\mathbb{E}[\cos^{2}(2\theta_{j})]=\frac{1}{2}$ as described above, as well as noticing that since each $\theta_{j}$ are independent, each $H_{j}\in \mathcal{C}_{(a,b)}$ contributes a factor of $\frac{1}{2}$ in the expectation except for $H_{k}$, as the values corresponding to $H_{k}$ are cancelled out.
From \eqref{eq:appVar2} to \eqref{eq:appVar3}, we applied simple counting and re-arranging, yielding the relatively simple form in \eqref{eq:appVar3}.

As such, whether the variance of the gradient vanishes exponentially depends on the quantity $\abs{\mathcal{K}_{(a,b)}}/2^{\abs{\mathcal{C}_{(a,b)}}}$.
For the classical ansatz \eqref{eq:classical}, we see that $\abs{\mathcal{C}_{(a,b)}}=2$ and $\abs{\mathcal{K}_{(a,b)}}=1$, so that the variance is simply $\sum_{(a,b)\in \mathsf{Cut}(H_{k})}w_{a,b}^{2} > \mathcal{O}(1)$. Consequently, the classical ansatz \eqref{eq:classical} does not exhibit barren plateaus.
For arbitrary $\mathbb{X}$-ans\"atze, it remains an open question how the scaling of $\abs{\mathcal{K}_{(a,b)}}$ compares to that of $2^{\abs{\mathcal{C}_{(a,b)}}}$.
We remark that while the quantity $\abs{\mathcal{C}_{(a,b)}}$ can be computed in linear time, determining $\abs{\mathcal{K}_{(a,b)}}$ constitutes a bottleneck in evaluating the objective function \eqref{eq:appcost}.
In fact, determining the set $\mathcal{K}_{(a,b)}$ can be shown to be \textsf{NP}-hard: for any $\mathbb{X}$-ansatz element $H_{j}=\bigotimes_{i\in S_{j}}X_{i}$ for $S_{j}\subset V$, represent $H_{j}$ as an $n$-bit binary string $\bm{x}_{j}$ with $x_{j,k}=1$ if and only if $k\in S_{j}$.
Then, the condition $\oplus\{H\in K\}=\emptyset$ is equivalent to the condition $\oplus\{\bm{x}_{j} \in K\}=\bm{0}$ for some subset $K\subset \mathcal{C}_{(a,b)}$.
This is a well-known problem referred to as computing the minimum distance of a binary linear code, which was shown to be \textsf{NP}-hard in \cite{Vardy1997}, thus showing that determining $\mathcal{K}_{(a,b)}$ is also \textsf{NP}-hard.

Thus, we see that determining the existence of barren plateaus in $\mathbb{X}$-ans\"atze is related to whether the objective function can be evaluated efficiently on a classical computer.

\section{Comparison between the GW algorithm and BFGS for solving MaxCut using the ansatz \eqref{eq:classical} } \label{App:ClassicalNumerics}
Here, we numerically compare the effectiveness of the classical ansatz \eqref{eq:classical}, consisting of local rotations around $X$ only, and the GW algorithm for solving MaxCut. In particular, we compare the approximation ratios $\alpha_{grad}$ obtained from solving \eqref{eq:maxcutquant} using BFGS and the classical ansatz \eqref{eq:classical} against the approximation ratios $\alpha_{GW}$ obtained from the GW algorithm. In all simulations we used the scipy optimizer L-BFGS-B \cite{BFGS} with hyperparameters $\text{gtol}=10^{-6}$ and $\text{ftol}=10^{-5}$. We used the GW algorithm implemented in the python package cvxgraphalgs \cite{GWweb}.

Fig. \ref{fig:classical} shows the ratio $\alpha_{grad}/\alpha_{GW}$ as a function of the vertex degree of a randomly chosen $k$-regular graph with $n=30$ (red), $n=50$ (blue), and $n=70$ (green) vertices, for edge weights $w_{a,b}\in[0,5]$.

\begin{figure}[h!]
	\includegraphics[width=0.4\linewidth]{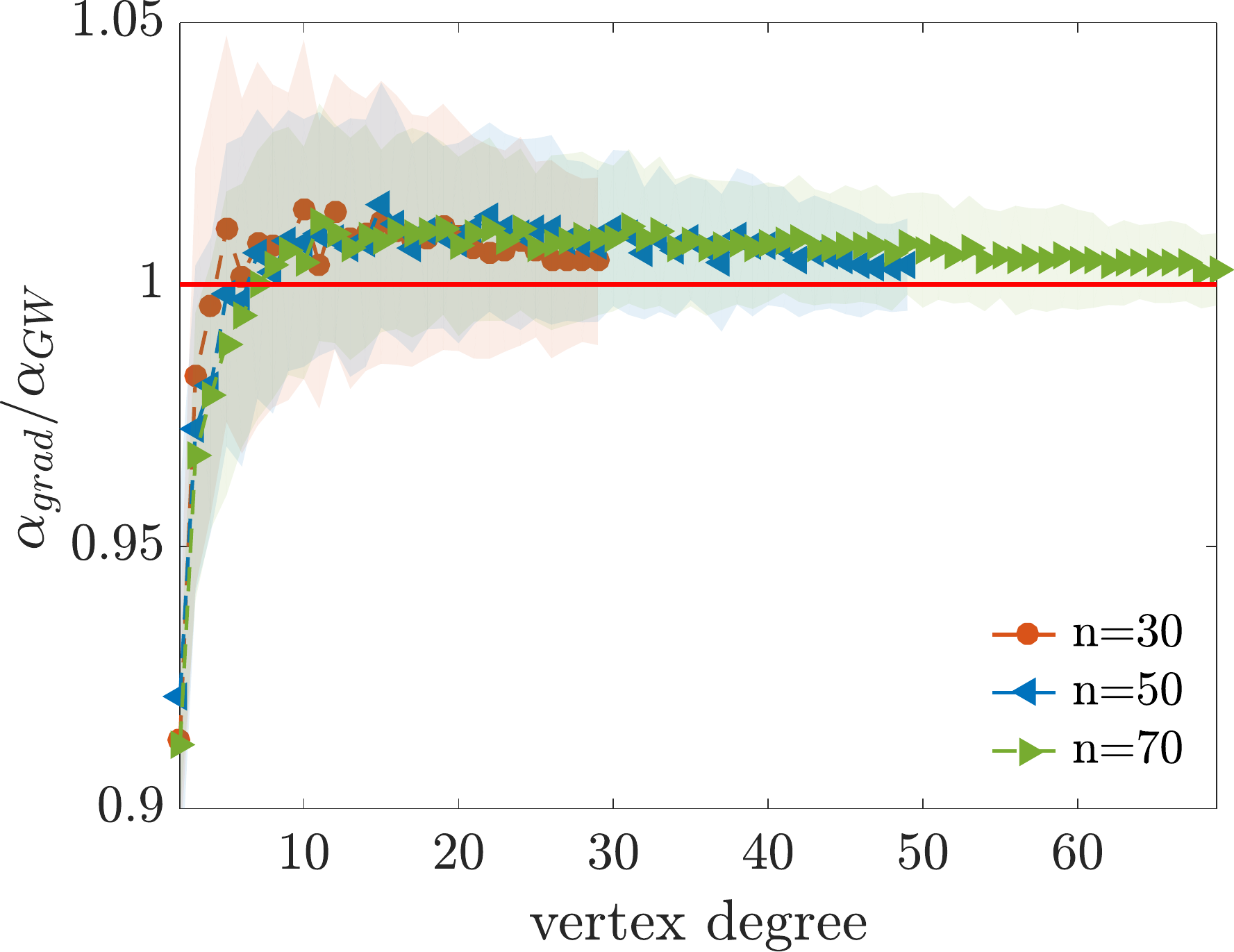}
	\caption{Comparison between the performance of the classical ansatz \eqref{eq:classical} coupled with a BFGS algorithm and the performance of the GW algorithm. The ratio between the two corresponding approximation ratios $\alpha_{grad}$ and $\alpha_{GW}$ is shown as a function of the vertex degree of a $k$-regular graph for $n=30,50,70$. Each data point shows the average value of $\alpha_{grad}/\alpha_{GW}$ taken over 200 different graph realizations, and the associated shaded area shows the standard deviation. The straight red line indicates when BFGS and \eqref{eq:classical} performs better, on average, than the GW algorithm.}    
	\label{fig:classical}
\end{figure}

\end{document}